\documentclass[11pt]{article}
\usepackage[margin=1in]{geometry}
\usepackage[T1]{fontenc}
\usepackage{lmodern,microtype,amsmath,amssymb,amsthm,mathtools,bm}
\usepackage{booktabs,array,longtable,enumitem,xcolor}
\newcolumntype{P}[1]{>{\raggedright\arraybackslash}p{#1}}
\usepackage[colorlinks=true,linkcolor=blue!45!black,citecolor=blue!45!black,urlcolor=blue!55!black]{hyperref}
\hypersetup{pdftitle={Minimax Quantum State Tomography in a Known Trace-Distance Ball},pdfsubject={Local minimax estimation with collective quantum measurements},pdfkeywords={quantum state tomography, trace distance, minimax risk, sample complexity}}
\usepackage{fancyhdr}
\setlist{topsep=4pt,itemsep=3pt}
\newtheorem{theorem}{Theorem}[section]
\newtheorem{lemma}[theorem]{Lemma}
\newtheorem{proposition}[theorem]{Proposition}
\newtheorem{corollary}[theorem]{Corollary}
\theoremstyle{definition}
\theoremstyle{remark}\newtheorem{remark}[theorem]{Remark}
\newcommand{\Tr}{\operatorname{Tr}}
\newcommand{\E}{\mathbb E}
\newcommand{\Prb}{\mathbb P}
\newcommand{\cD}{\mathcal D}
\newcommand{\cB}{\mathcal B}
\newcommand{\cR}{\mathcal R}
\newcommand{\cQ}{\mathcal Q}
\newcommand{\norm}[1]{\left\lVert #1\right\rVert}
\newcommand{\ket}[1]{\lvert #1\rangle}
\newcommand{\bra}[1]{\langle #1\rvert}
\newcommand{\ip}[2]{\langle #1,#2\rangle}
\newcommand{\op}{\mathrm{op}}
\newcommand{\HS}{\mathrm{HS}}
\newcommand{\vecop}{\operatorname{vec}}
\DeclareMathOperator{\diag}{diag}
\DeclareMathOperator{\GammaDist}{Gamma}

\title{Minimax Quantum State Tomography\\in a Known Trace-Distance Ball}
\author{Yonglong Li\\School of Mathematics and Statistics\\ Xi'an Jiaotong University\\ liyonglong@xjtu.edu.cn}
\date{}
\begin{document}
\maketitle
\begin{abstract}
We study quantum state tomography when the unknown state lies in a trace-distance ball of radius $R$ around a known reference state $\rho$. For arbitrary collective measurements, we determine the minimax expected risk and confidence radius up to universal constants, uniformly over the dimension, center spectrum, radius, and sample size. Write $\Gamma=d(\Tr\sqrt\rho)^2+d^2R$. For $d\ge2$, $0<R\le1$, and $n\ge1$ copies, the expected risk is $\Theta(\min\{R,\sqrt{\Gamma/n}\})$. At failure probability $0<\delta\le1/3$, the minimax confidence radius is
\[
 \Theta\!\left(\min\left\{R,\sqrt{\frac{\Gamma+\log(1/\delta)}{n}}\right\}\right).
\]
For $\varepsilon\le c_0R$ with a universal $c_0>0$, the copy complexity is $\Theta((\Gamma+\log(1/\delta))/\varepsilon^2)$. For the lower bound we use a weighted trace-norm witness to retain all spectral scales in a square-root perturbation family. For the upper bound, we give a direct spectral and concentration analysis of a modified existing purification-based procedure.
\end{abstract}
\noindent\textbf{Index terms:} Quantum state tomography, local minimax estimation, trace distance, collective measurements, sample complexity.

\section{Introduction}\label{sec:introduction}

Suppose that a preparation procedure is designed to produce a known quantum state $\rho$, and that calibration or a model gives the deterministic promise
\begin{equation}\label{eq:introball}
 T(\sigma,\rho)\le R
\end{equation}
for the state $\sigma$ actually produced. Here $T(\sigma,\rho)=\frac12\norm{\sigma-\rho}_1$ is the trace distance, which controls the difference between the probabilities of every measurement event~\cite{Watrous}. The task is to reconstruct $\sigma$ from independent copies. The promise is assumed as part of the model; verifying it is a separate statistical problem. 

Returning $\rho$ gives error at most $R$ without using any copies. To improve on this estimate, one must learn the variations allowed by the ball. These variations depend on more than the rank of the center. Even a small ball around a pure state contains both coherent changes of direction and mixed states with mass on previously unoccupied directions. We ask for the minimax risk uniformly over the dimension, the entire spectrum of the center, the radius, and the sample size.

For unrestricted collective measurements, global trace-distance tomography has copy complexity $\Theta(d^2/\varepsilon^2)$ at constant confidence, and the corresponding rank-at-most-$r$ rate is $\Theta(rd/\varepsilon^2)$~\cite{OW,HHJWY,SSW}. These results do not by themselves give a matching lower bound in an arbitrary prescribed neighborhood. Our main result identifies its coefficient as
\begin{equation}\label{eq:introGamma}
 \Gamma(\rho,R)=d(\Tr\sqrt\rho)^2+d^2R.
\end{equation}
The first term retains all spectral scales of the center, while the second accounts for additional states admitted by the radius. Theorem~\ref{thm:main} gives the expected risk and the confidence radius for every sample size; Corollary~\ref{thm:confidence} gives the resulting high-accuracy copy complexity.

The main contribution is the matching local converse for arbitrary spectra, including zero and arbitrarily small eigenvalues. We regularize the center within the ball and perturb its square root by a full family of sign matrices. Neighboring-state tests force many weighted sign errors. A trace-norm dual witness built from the true sign matrix then converts these errors into the full sum $\Tr\sqrt\rho$, without selecting a single spectral scale. This step applies the dual witness of Acharya \emph{et al.}~\cite[Lemma 4.7]{ADLY}, the true sign pattern, to a spectrally weighted square-root family. The argument yields a constant-probability error event uniformly in all parameters. A local two-point construction supplies the additional confidence term at every center.

For achievability, we analyze the existing purification/Hayashi procedure~\cite{TWZ,PSTW,Hayashi}. Its measurement uses neither the center nor its spectrum. An exact Gaussian representation gives a direct expected-loss bound and a tail bound in which $\log(1/\delta)$ enters additively with $\Gamma$. Combined with the option of returning $\rho$, this one measurement is rate-optimal over every prescribed ball. The theorem thus identifies the difficulty of the promise class; it does not require a new center-dependent measurement. The upper rates themselves are implied by earlier results. The expected spectral rate follows from the second moments of this procedure~\cite[Theorem 1.7]{PSTW}, and the confidence rate follows from its observable concentration~\cite[Proposition 4.7]{PSTWspectrum} by generic chaining~\cite{Dirksen}. These trace-distance bounds are not stated in the papers of Pelecanos \emph{et al.}, so Section~\ref{sec:upper} gives a short direct proof with explicit constants; Section~\ref{sec:comparison} and Remark~\ref{rem:priorupper} separate the inherited ingredients from the local converse.

In Section~\ref{sec:model} we state the results, their consequences, and the comparison with earlier guarantees. Section~\ref{sec:facts} collects the required standard facts. In Sections~\ref{sec:lower}--\ref{sec:confidenceLower} we prove the lower bound, upper bound, and confidence refinement. In the appendices we state two analytic tools and the fixed-rank rate at a normalized-projector center, which is implied by existing results.

\section{Problem setting and main results}\label{sec:model}

\subsection{Statistical experiment and loss}

Let $d\ge2$ and write
\[
 \cD_d=\{\sigma\in\mathbb C^{d\times d}:\sigma\succeq0,\ \Tr\sigma=1\},
 \qquad T(\sigma,\tau)=\tfrac12\norm{\sigma-\tau}_1.
\]
The trace norm is $\norm X_1=\Tr\sqrt{X^\dagger X}$, the Hilbert--Schmidt norm is $\norm X_{\HS}=\sqrt{\Tr X^\dagger X}$, and $\norm X_{\op}$ denotes the operator norm. We reserve $\norm{\cdot}_2$ for the Euclidean norm of a vector.

The estimator knows a classical description of $\rho\in\cD_d$, including its eigenvectors, and a radius $R\in(0,1]$. The unknown state belongs to
\[
 \cB_R(\rho)=\{\sigma\in\cD_d:T(\sigma,\rho)\le R\}.
\]
An admissible procedure applies any POVM to $\sigma^{\otimes n}$ and outputs a density matrix $\widehat\sigma$. It may use $(d,n,\rho,R)$, the requested confidence or accuracy when specified, and private randomness. It cannot use the unknown $\sigma$. We impose no computational or restricted-measurement requirement.

Define the minimax expected risk and confidence radius by
\[
 \cR_n(\rho,R)=\inf_{\mathsf A}\sup_{\sigma\in\cB_R(\rho)}
       \E_\sigma T(\widehat\sigma_{\mathsf A},\sigma),
\]
\[
 \cQ_n(\rho,R,\delta)=\inf_{\mathsf A}\inf\left\{u\ge0:
 \sup_{\sigma\in\cB_R(\rho)}\Prb_\sigma
 \{T(\widehat\sigma_{\mathsf A},\sigma)>u\}\le\delta\right\}.
\]
Let $N(\rho,R,\varepsilon,\delta)$ be the least nonnegative integer $n$ for which an admissible procedure has error at most $\varepsilon$ with failure probability at most $\delta$, uniformly over the ball. Set
\begin{equation}\label{eq:Gamma}
 \begin{split}
 S(\rho)&=\Tr\sqrt\rho,\qquad \Gamma(\rho,R)=dS(\rho)^2+d^2R,\\
 \Gamma_\delta&=\Gamma(\rho,R)+\log(1/\delta),\\
 q_n&=\min\{R,\sqrt{\Gamma(\rho,R)/n}\},\qquad
 q_{n,\delta}=\min\{R,\sqrt{\Gamma_\delta/n}\}.
 \end{split}
\end{equation}
Since $1\le S(\rho)^2\le d$, we have $d\le\Gamma(\rho,R)\le2d^2$. All logarithms are natural. The notation $a\asymp b$ means that $cb\le a\le Cb$ for fixed positive numerical constants independent of all problem parameters.

\begin{theorem}[Local minimax risk and confidence radius]\label{thm:main}
There are universal constants $c,C>0$ such that, for every $d\ge2$, $\rho\in\cD_d$, $0<R\le1$, integer $n\ge1$, and $0<\delta\le1/3$,
\[
 c q_n\le\cR_n(\rho,R)\le Cq_n,
\]
\begin{equation}\label{eq:fullconfidence}
 c q_{n,\delta}\le\cQ_n(\rho,R,\delta)\le Cq_{n,\delta}.
\end{equation}
In particular, $\cQ_n(\rho,R,1/3)\asymp q_n$. The lower-bound construction also shows that, for every procedure, some $\sigma\in\cB_R(\rho)$ satisfies
\begin{equation}\label{eq:stronglower}
 \Prb_\sigma\{T(\widehat\sigma,\sigma)>10^{-6}q_n\}>1/3.
\end{equation}
\end{theorem}

The theorem holds without a rank, spectral-gap, or lower-eigenvalue assumption. Equation~\eqref{eq:stronglower} supplies the expected-risk and constant-confidence converses directly. The numerical constant is conservative; the theorem concerns uniform parameter dependence rather than sharp leading constants.

\begin{corollary}[Copy complexity in the high-accuracy regime]\label{thm:confidence}
There are universal constants $c_0,c,C>0$ such that, whenever $d\ge2$, $\rho\in\cD_d$, $0<R\le1$, $0<\varepsilon\le c_0R$, and $0<\delta\le1/3$,
\begin{equation}\label{eq:copy}
 c\frac{\Gamma(\rho,R)+\log(1/\delta)}{\varepsilon^2}
 \le N(\rho,R,\varepsilon,\delta)
 \le\left\lceil C\frac{\Gamma(\rho,R)+\log(1/\delta)}{\varepsilon^2}\right\rceil.
\end{equation}
\end{corollary}

The proof of this corollary is given after the confidence converse in Section~\ref{sec:confidenceLower}.

\begin{remark}[Boundary cases and accuracy]
If $R=0$ or $d=1$, the unknown state is known. If $\varepsilon\ge R$, returning $\rho$ succeeds without copies. Theorem~\ref{thm:main} describes the confidence radius for all $n\ge1$ and $0<\delta\le1/3$, up to universal constants. Inverting such a bound gives Corollary~\ref{thm:confidence} only below a sufficiently small multiple of $R$. The small value of $c_0$ furnished by our proof reflects conservative lower-bound constants; it is not a statistical transition. Neither statement locates the exact copy-complexity transition as $\varepsilon/R$ approaches one. All results concern arbitrary collective measurements on the unrestricted ball.
\end{remark}

\subsection{Consequences and representative centers}\label{sec:interpretation}

The spectral size $S(\rho)^2$ equals one for a pure state and $r$ for a normalized rank-$r$ projector. More generally, it retains the weights of small eigenvalues rather than only their number. At constant confidence, Corollary~\ref{thm:confidence} gives the following examples.

\begin{center}
\begin{tabular}{P{.42\linewidth}P{.46\linewidth}}
\toprule
Known center & Copies for $\varepsilon\le c_0R$\\
\midrule
Pure state $P$ & $\asymp(d+d^2R)/\varepsilon^2$\\
Maximally mixed state $I/d$ & $\asymp d^2/\varepsilon^2$\\
Normalized projector $P_r/r$ & $\asymp(dr+d^2R)/\varepsilon^2$\\
$\diag(1-u,u/(d-1),\ldots)$, $0\le u\le1/2$ & $\asymp[d+d^2(u+R)]/\varepsilon^2$\\
$\lambda_j=j^{-2}/\sum_{i=1}^di^{-2}$ & $\asymp[d(\log(d+1))^2+d^2R]/\varepsilon^2$\\
\bottomrule
\end{tabular}
\end{center}

For the peaked spectrum, $S(\rho)^2=(\sqrt{1-u}+\sqrt{(d-1)u})^2\asymp1+du$. For the power-law spectrum, the normalizing sum stays bounded above and below, while $\sum_{j=1}^d1/j\asymp\log(d+1)$. Thus exact rank alone can greatly overestimate the coefficient in the local rate.

Let $N_{\mathrm{global}}(d,\varepsilon,\delta)$ denote copy complexity over all $d$-dimensional states. A radius-one ball is the entire state space, so Corollary~\ref{thm:confidence} also gives $N_{\mathrm{global}}\asymp(d^2+\log(1/\delta))/\varepsilon^2$ in the common high-accuracy regime. Consequently,
\begin{equation}\label{eq:savingsconfidence}
 \frac{N(\rho,R,\varepsilon,\delta)}
      {N_{\mathrm{global}}(d,\varepsilon,\delta)}
 \asymp
 \frac{dS(\rho)^2+d^2R+\log(1/\delta)}{d^2+\log(1/\delta)}.
\end{equation}
Along a sequence of problems, this ratio tends to zero if and only if
\begin{equation}\label{eq:savingsconditions}
 R\longrightarrow0,\qquad
 \frac{S(\rho)^2}{d}\longrightarrow0,\qquad
 \frac{\log(1/\delta)}{d^2}\longrightarrow0.
\end{equation}
Indeed, a vanishing ratio first forces $\log(1/\delta)=o(d^2)$, after which the remaining two conditions follow by dividing the numerator by $d^2$. The converse is immediate from the same expression. At constant confidence, the last condition is automatic whenever $S(\rho)^2/d\to0$. For a pure center, the saving factor is then $\asymp\min\{d,1/R\}$. If the confidence term dominates $d^2$, localization yields no unbounded saving.

At a pure center, the nonzero $d$ term comes from coherent changes of direction, even though the central spectrum is $(1,0,\ldots,0)$. The two-point rotation in Section~\ref{sec:confidenceLower} also explains why the confidence cost has no radius prefactor. These effects differ from the commuting boundary model in Remark~\ref{rem:classicalconfidence}.

The rank-$r$ row constrains only the center. If the unknown state also has rank $r$, Corollary~\ref{cor:fixedrank} shows, for $1\le r\le d/2$, that the expected risk and the constant-confidence radius over
\[
 \{\sigma\in\cD_d:\operatorname{rank}\sigma=r,
                 \ T(\sigma,P_r/r)\le R\}
\]
are both $\asymp\min\{R,\sqrt{rd/n}\}$. The same statement holds for rank at most $r$. Thus admitting rank-changing alternatives adds the $d^2R$ contribution at this center. This fixed-rank rate is implied by existing results: the lower bound follows from the proof of~\cite[Proposition 4.13]{NZ}, and the upper bound from known-rank estimators~\cite{OW,PSTW}. Appendix~\ref{app:rank} gives the deduction together with a short direct proof. The corollary does not characterize neighborhoods of arbitrary nonuniform fixed-rank centers.

\subsection{Comparison with existing guarantees}\label{sec:comparison}

The nearest results address different parts of the problem. Table~\ref{tab:priorwork} identifies which ingredients are inherited and which local dependence still requires a converse.

\begin{table}[htbp]
\centering\small
\caption{Relation to existing guarantees and techniques.}\label{tab:priorwork}
\begin{tabular}{P{.22\linewidth}P{.34\linewidth}P{.34\linewidth}}
\toprule
Results & Established guarantee & Relation to the present theorem\\
\midrule
Global and low-rank tomography~\cite{OW,HHJWY,SSW}
& Optimal global rank-dependent rates; approximate-low-rank upper bounds.
& A prescribed ball need not admit a useful low-rank approximation at the requested accuracy.\\
Local lower bounds~\cite{NZ,KLMR}
& Well-conditioned centers or localized projector families, including bounded-batch models.
& The projector family of~\cite{NZ} already gives the fixed-rank local lower bound of Corollary~\ref{cor:fixedrank}. Neither gives the arbitrary-center coefficient $dS(\rho)^2+d^2R$.\\
Trace-distance Hamming separation~\cite{ADLY}
& Dual witness built from the true sign parameter, for an unweighted family at the maximally mixed state.
& Our converse applies this witness to a spectrally weighted square-root family.\\
Purification-based estimators~\cite{PSTW,PSTWspectrum}
& Known-rank guarantees; second moments and observable concentration for an existing estimator.
& With generic chaining~\cite{Dirksen}, these imply the local upper rates in expectation and at every confidence level (Remark~\ref{rem:priorupper}); the trace-distance bounds are not stated there.\\
Present work
& Matching local converse for arbitrary spectra and a uniform confidence-radius formula.
& Singular centers, rank changes, and confidence cost are covered by one finite-sample rate; the upper rates are reproved directly with explicit constants.\\
\bottomrule
\end{tabular}
\end{table}

\paragraph{Tomography and local lower bounds.}
O'Donnell and Wright~\cite{OW} prove the $O(rd/\varepsilon^2)$ trace-distance upper bound for rank-at-most-$r$ states and guarantees based on low-rank approximation. Haah \emph{et al.}~\cite{HHJWY} prove the tight full-dimensional lower bound; Scharnhorst, Spilecki, and Wright~\cite{SSW} prove the matching rank-dependent lower bound, even under the stronger normalized-projector promise in the low-rank regime. Our global specializations provide alternative proofs of these known rates.

In concurrent work, Keskin \emph{et al.}~\cite[Corollary 1.4]{KLMR} consider centers $\rho\succeq I/(2d)$ and neighborhoods $\norm{\sigma-\rho}_{\op}\le L\varepsilon/d$, for a universal constant $L\ge4$. Their result gives an $\Omega(d^2/\varepsilon^2)$ lower bound for collective tomography in $\cB_R(\rho)$ when $\varepsilon\le2R/L$, since $T\le(d/2)\norm{\cdot}_{\op}$, in their small-accuracy regime. Nayak and Zhou~\cite{NZ} obtain the optimal bounded-batch rank-dependent rate. For $r\le d/2$, their hard family consists of normalized projectors close to a fixed one, so the proof of their Proposition 4.13 already yields the fixed-rank local lower rate of Corollary~\ref{cor:fixedrank}, although no local statement is formulated there; Appendix~\ref{app:rank} gives the deduction. Their proof uses Fisher information and a smoothly truncated Gaussian density, whereas ours uses a discrete family and weighted trace-norm decoding. Neither result states the arbitrary-center coefficient in~\eqref{eq:introGamma}.

\paragraph{Antecedents of the converse.}
The decoding step in Section~\ref{sec:dualwitness} extends the trace-distance Hamming separation of Acharya \emph{et al.}~\cite[Lemma 4.7]{ADLY}. For an unweighted sign family around the maximally mixed state, they test $\sigma_z-\sigma_{z'}$ against the perturbation of $z$ alone, so that only the true parameter needs an operator-norm bound. Their setting is single-copy measurements, but this step is algebraic. We add spectral weights and a square-root parametrization, which retain $\Tr\sqrt\rho$ at an arbitrary center.

The radius term alone also follows from the global bound~\cite{HHJWY}. The states $(1-R)\rho+R\tau$ lie in $\cB_R(\rho)$, and $n$ copies of one of them can be prepared from $\rho$ and, with high probability, $O(nR+1)$ copies of $\tau$. An estimate within $\varepsilon$ of such a state yields, after rescaling and projection onto states, an estimate of $\tau$ within $2\varepsilon/R$; this gives $\Omega(d^2R/\varepsilon^2)$ for $\varepsilon\le cR$ at constant confidence. With the projector family of~\cite{NZ}, this reduction recovers the constant-confidence converse at normalized-projector centers with $r\le d/2$. The spectral term $dS(\rho)^2$ at a general center is not covered by these reductions.

\paragraph{A comparison with low-rank approximation guarantees.}
Let $\lambda_j=j^{-2}/\sum_{i=1}^d i^{-2}$, take
$R_d=(\log(d+1))^2/d$, and fix $\varepsilon_d=cR_d$ with
$0<c<c_0$. For all sufficiently large $d$, the radius is at most one,
and Corollary~\ref{thm:confidence} gives, at constant confidence,
\[
 N(\rho,R_d,\varepsilon_d,1/3)
 \asymp\frac{d(\log(d+1))^2}{\varepsilon_d^2}.
\]
This conclusion cannot be obtained by directly applying an approximate
rank promise throughout the ball. Indeed, the ball contains
$\sigma_d=(1-R_d)\rho+R_dI/d$, and its ordered eigenvalues satisfy
\[
 \sum_{j>r}\lambda_j(\sigma_d)\ge R_d(1-r/d).
\]
For any state $\tau$ of rank at most $r$, measuring the projector onto
its support shows that
$T(\sigma_d,\tau)\ge\sum_{j>r}\lambda_j(\sigma_d)$.
Consequently, approximation to accuracy $\varepsilon_d$ requires
$r\ge(1-c)d$, even though the center has a rapidly decaying spectrum.
The approximate-rank guarantee of O'Donnell and Wright~\cite[Theorem 1.5
and Corollary 1.6]{OW} uses $O(rd/\eta^2)$ copies for trace-distance
error at most the optimal rank-$r$ approximation error plus $\eta$.
Optimizing this guarantee uniformly over the ball therefore retains
the scale $d^2/\varepsilon_d^2$, which exceeds the local minimax scale
by a factor of order $d/(\log(d+1))^2$.
This comparison concerns the stated approximate-rank guarantee, not a
lower bound on the performance of its underlying full-rank estimator.
The spectral upper bounds implied by~\cite{PSTW,PSTWspectrum} do recover
the local upper rate after localization; our
matching converse establishes that rate as the minimax answer.

\paragraph{Upper-bound ingredients.}
The debiased estimator of Pelecanos, Spilecki, and Wright~\cite{PSW} has spectral second-moment guarantees. For the estimator $\mathrm{Mix}(\mathcal A_{\mathrm{GPS}})$ with a $d$-dimensional ancilla, which uses the measurement analyzed here, existing results imply
\begin{equation}\label{eq:priorupper}
 T(\widehat\sigma,\sigma)
 =O\!\left(\sqrt{d/n}\,S(\sigma)+d^2/n\right)
\end{equation}
after an affine change of the output: in expectation by~\cite[Theorem 1.7]{PSTW}; at constant confidence by the proof of Theorem 4.8 in~\cite{PSTWspectrum}; and, with the additional terms $\sqrt{\log(1/\delta)/n}+\log(1/\delta)/n$, with probability at least $1-\delta$ by~\cite[Proposition 4.7]{PSTWspectrum} and generic chaining~\cite[Theorem 3.5]{Dirksen}. These trace-distance bounds are not stated in the cited papers; Remark~\ref{rem:priorupper} gives the deductions, the last one in outline. After localization they yield the upper halves of Theorem~\ref{thm:main}, and the new converse establishes their local optimality.

For the purification/Hayashi procedure, Proposition 1.4 of~\cite{PSTW} gives the pure-state guarantee. Combining it with their mixed-to-pure reduction, as noted immediately after that proposition, gives infidelity at most $\eta$ with $O((rd+\log(1/\delta))/\eta)$ copies using an $r$-dimensional ancilla. Since $T\le\sqrt{1-F}$, setting $\eta=\varepsilon^2$ gives the known-rank trace-distance bound. Their Theorem 1.3 itself uses the gate-efficient pure-state algorithm $\mathcal A_{\mathrm{GKKT}}$. We analyze the purification/Hayashi procedure with ancilla dimension $d$ directly, through an exact Gaussian representation; this retains the full spectral sum and separates the confidence fluctuation with explicit constants. Its exact bias, a consequence of the unbiasedness of the debiased estimator~\cite[Lemma 5.7]{PSTW}, and the role of the ancilla dimension are stated in Proposition~\ref{prop:bias} and Remark~\ref{rem:ancilladimension}.

\paragraph{Other local and instance-dependent formulations.}
Quantum local asymptotic normality and minimax estimation~\cite{KG,GJK} concern regular finite-dimensional local models. Lahiry and Nussbaum~\cite{LN} study fixed-rank shrinking trace-norm neighborhoods, with squared Hilbert--Schmidt loss at fixed dimension and rank and distinct positive nonzero eigenvalues. These asymptotic statements do not supply the present uniform trace-distance result.

Certification asks whether the unknown state equals a target or is separated from it, rather than reconstructing every state in a neighborhood. Its worst-case copy complexity is $\Theta(d/\varepsilon^2)$~\cite{OWtesting,BOW}, and the instance-dependent bounds of O'Donnell and Wadhwa~\cite{OWcert} already involve
\begin{equation}\label{eq:certfunctional}
 dF(\rho,I/d)=(\Tr\sqrt\rho)^2,
\end{equation}
with spectral truncations and logarithmic factors. Here $F$ is squared fidelity. Thus the square-root spectral functional itself is not new. Classical instance-dependent learning and testing~\cite{VVlearning,VVtesting} provide related motivation: the empirical distribution satisfies $\E\norm{\widehat p-p}_1\le n^{-1/2}\sum_i\sqrt{p_i(1-p_i)}$, with sharper coordinatewise bounds available~\cite{BK}. The present problem additionally includes quantum coherences and rank-changing neighborhoods.

\subsection{Proof strategy}

The converse must capture both the center's spectral contribution and the mass that the ball permits in small or unoccupied directions. Replace the center temporarily by $\rho_0=(1-R/2)\rho+(R/2)I/d$. This uses at most half the radius and gives a spectral floor $R/(2d)$. A rectangular sign perturbation of $\sqrt{\rho_0}$ produces positive states after normalization, and the entire sign cube remains in the ball. Flipping one sign changes a purification by only order $n^{-1/2}$, so neighboring-state tests force many sign errors. Testing the off-diagonal error block against the true sign matrix, as in the unweighted argument of~\cite{ADLY}, converts these errors into a weighted sum with weights $\sqrt{\lambda_i(\rho_0)}$. The spectral floor controls the quadratic normalization cost.

The power-law center above illustrates why the weights must be retained
simultaneously. Let $\lambda_i=Z_d^{-1}i^{-2}$, where
$Z_d=\sum_{i=1}^d i^{-2}$, and divide the first
$k=\lfloor d/2\rfloor$ indices into dyadic bands
$B_j=\{2^j,\ldots,2^{j+1}-1\}$, keeping the complete bands.
Since $1\le Z_d\le2$, we have
\[
 \sum_{i\in B_j}\sqrt{\lambda_i}\asymp1,
 \qquad
 \sum_{i=1}^k\sqrt{\lambda_i}\asymp\log(d+1).
\]
Thus, in a rectangular perturbation with $m\asymp d$ rows,
the spectral coefficient
$m(\sum_{i\in B_j}\sqrt{\lambda_i})^2$ obtained from any one
band is only of order $d$, whereas retaining all bands gives
$d(\log(d+1))^2$. The same comparison applies after our
regularization when $R\le d^{-2}$, because then
$\mu_i=(1-R/2)\lambda_i+R/(2d)\asymp\lambda_i$ uniformly in
$i$. This comparison does not exclude arguments that combine
several bands; it shows why selecting a single favorable band
can miss the accumulation across spectral scales. The weighted
witness in Section~\ref{sec:dualwitness} preserves this accumulation
within one sign family.

For achievability, condition on the common random purification produced by the existing channel. The covariant measurement has an exact normalized Gaussian representation. Its reduced-state error splits into a linear term with mean $O(\sqrt d\,S(\sigma))$ and a quadratic remainder of order $d^2/n$. The linear term has a dimension-independent Lipschitz constant, so Gaussian concentration adds $\sqrt{\log(1/\delta)/n}$ without multiplying the spectral coefficient. The radius contribution absorbs the quadratic remainder whenever using the measurement improves on returning the center.

For the confidence converse, a balanced diagonal perturbation or a coherent two-dimensional rotation produces two states in the same ball with separation of order $\varepsilon$ and fidelity $1-O(\varepsilon^2)$. Binary testing then forces the confidence contribution at every center. Combining it with the weighted sign-family converse gives~\eqref{eq:fullconfidence} for every $n\ge1$.

\section{Preliminaries}\label{sec:facts}

\subsection{Matrix norms and off-diagonal blocks}

We use three standard properties of Schatten norms; see, for example,
Bhatia~\cite{Bhatia}. For a rectangular matrix $B$, nuclear/operator norm
duality gives
\begin{equation}\label{eq:duality}
 \norm B_1=\sup_{\norm C_{\op}\le1}|\Tr(C^\dagger B)|.
\end{equation}
Since a rank-one matrix $uv^\dagger$ has trace norm
$\norm u_2\norm v_2$, decomposing $B$ into its row matrices gives
\begin{equation}\label{eq:rownorm}
 \norm B_1\le\sum_i\norm{B_{i,:}}_2.
\end{equation}
The Schatten Cauchy--Schwarz inequality states that
\begin{equation}\label{eq:holder}
 \norm{UV^\dagger}_1\le\norm U_{\HS}\norm V_{\HS}.
\end{equation}
The following consequence lets us use an off-diagonal block as a
trace-distance witness.

\begin{lemma}[Extracting an off-diagonal block]\label{lem:block}
For a Hermitian block matrix $H=\begin{pmatrix}H_{11}&B\\B^\dagger&H_{22}\end{pmatrix}$,
\[
 \tfrac12\norm H_1\ge\norm B_1.
\]
\end{lemma}
\begin{proof}
Let $J=\diag(I,-I)$. Since
\[
 \frac{H-JHJ}{2}=\begin{pmatrix}0&B\\B^\dagger&0\end{pmatrix}
\]
has nonzero eigenvalues equal to the positive and negative singular
values of $B$, the triangle inequality and unitary invariance give
\[
 2\norm B_1=\norm{(H-JHJ)/2}_1
 \le\tfrac12(\norm H_1+\norm{JHJ}_1)=\norm H_1,
\]
thus dividing by two proves the desired claim.
\end{proof}

\subsection{Purifications and binary testing}

For a matrix $A=(A_{ij})$, define
$\vecop(A)=\sum_{ij}A_{ij}\ket i\otimes\ket j$. The identities
\[
 \ip{\vecop(A)}{\vecop(B)}=\Tr A^\dagger B,
 \qquad \Tr_2\ket{\vecop(A)}\bra{\vecop(A)}=AA^\dagger
\]
show that $\vecop(A)$ purifies $AA^\dagger$ whenever $\norm A_{\HS}=1$.

We also use the following standard facts~\cite{Watrous}. Trace distance
cannot increase under a quantum channel, and for unit vectors $u,v$,
\begin{equation}\label{eq:pureT}
 T(\ket u\bra u,\ket v\bra v)=\sqrt{1-|\ip u v|^2}.
\end{equation}
Thus any chosen purifications give an upper bound on the trace distance
of their reduced states. The minimum equal-prior error of a binary
test, allowing randomized tests, is
\begin{equation}\label{eq:helstrom}
 P_e^*(\sigma,\tau)=\tfrac12(1-T(\sigma,\tau)).
\end{equation}
We use the squared-fidelity convention
\begin{equation}\label{eq:fidelity}
 F(\sigma,\tau)=\norm{\sqrt\sigma\sqrt\tau}_1^2.
\end{equation}
Since Uhlmann's theorem and~\eqref{eq:pureT} imply
$T(\sigma,\tau)\le\sqrt{1-F(\sigma,\tau)}$, and fidelity is
multiplicative under tensor products, from~\eqref{eq:helstrom} it
follows that
\begin{equation}\label{eq:testingF}
 P_e^*(\sigma^{\otimes n},\tau^{\otimes n})
 \ge\frac{1-\sqrt{1-F(\sigma,\tau)^n}}2.
\end{equation}

\section{A uniform lower bound from a square-root sign family}\label{sec:lower}

We prove the expected-risk and constant-confidence lower bounds in Theorem~\ref{thm:main}. Fix $\rho,R,d,n$, and consider an arbitrary measurement and estimator. We first construct a finite family of states in $\cB_R(\rho)$, then show that estimating these states requires recovering many signs that cannot be distinguished reliably from $n$ copies.

\subsection{Move the center slightly into the interior}

Choose an eigenbasis of $\rho$ and define
\begin{equation}\label{eq:regularize}
 \rho_0=(1-R/2)\rho+(R/2)I/d
       =\diag(\mu_1,\ldots,\mu_d),\qquad \mu_1\ge\cdots\ge\mu_d.
\end{equation}
Since $T(I/d,\rho)\le1$, it follows from~\eqref{eq:regularize} that
\[
 T(\rho_0,\rho)=(R/2)T(I/d,\rho)\le R/2,
 \qquad \mu_i\ge R/(2d).
\]
Thus a perturbation within distance $R/2$ of $\rho_0$ remains in the original ball, while every eigenvalue of $\rho_0$ has a positive lower bound. The state $\rho_0$ is used only to construct the lower-bound family.

Let
\begin{equation}\label{eq:L}
 k=\lfloor d/2\rfloor,\qquad m=d-k,\qquad
 s_0=\sum_{i=1}^k\sqrt{\mu_i},\qquad \ell_0=\sqrt m\,s_0.
\end{equation}
Since $\sqrt{\mu_1},\ldots,\sqrt{\mu_d}$ are nonincreasing, their first $k$ terms have average at least the average of all $d$ terms. Using $k\ge d/3$ and $m\ge d/2$, we then have
\[
 s_0\ge(k/d)S(\rho_0)\ge S(\rho_0)/3,
 \qquad \ell_0^2=ms_0^2\ge(d/18)S(\rho_0)^2.
\]
Moreover, $R\le1$ implies $\mu_i\ge\lambda_i(\rho)/2$, and the spectral floor gives $\sqrt{\mu_i}\ge\sqrt{R/(2d)}$. Summing these inequalities separately yields
\[
 S(\rho_0)^2\ge\tfrac12 S(\rho)^2,
 \qquad S(\rho_0)^2\ge\tfrac12 dR.
\]
Since the maximum of these two lower bounds is at least their average, it follows that
\begin{equation}\label{eq:Llower}
 \ell_0^2\ge\frac{dS(\rho)^2+d^2R}{72}=\frac\Gamma{72},
 \qquad \ell_0\ge\frac{\sqrt\Gamma}{9}.
\end{equation}
We shall also need to control the quadratic term in the perturbation. Since each of the first $k$ eigenvalues is at least $R/(2d)$, we have
\begin{equation}\label{eq:curvature}
 s_0^2\ge\frac{k^2R}{2d},\qquad \frac{k}{s_0^2}\le\frac{2d}{kR}\le\frac6R.
\end{equation}
Thus moving the center into the interior gives both the spectral size in~\eqref{eq:Llower} and the quadratic control in~\eqref{eq:curvature}.

\subsection{Construct every vertex of the finite family}

Let $\Xi$ range over the real $m\times k$ matrices whose entries belong to $\{-1,1\}$. The resulting family is a cube with $mk$ binary coordinates. With respect to $\mathbb C^d=\mathbb C^k\oplus\mathbb C^m$, define
\[
 D=\diag(\sqrt{\mu_1},\ldots,\sqrt{\mu_d}),
 \qquad E_\Xi=\begin{pmatrix}0&0\\\Xi&0\end{pmatrix},
\]
and put
\begin{equation}\label{eq:cube}
 a=\frac1{100}\min\{n^{-1/2},R/\ell_0\},\quad
 b=a^2km,\quad h=1+b,\quad
 A_\Xi=D+aE_\Xi,\quad \sigma_\Xi=\frac{A_{\Xi}A_\Xi^\dagger}{h}.
\end{equation}
Since $\Tr(DE_\Xi^\dagger)=0$ and $\norm{E_\Xi}_{\HS}^2=km$, we have
\[
 \Tr A_\Xi A_\Xi^\dagger
 =\Tr D^2+a\Tr(DE_\Xi^\dagger+E_\Xi D)
       +a^2\Tr E_\Xi E_\Xi^\dagger
 =1+a^2km=h.
\]
Consequently, $\sigma_\Xi$ is positive semidefinite with trace one, and the same normalizing constant applies to every vertex.

We next show that every $\sigma_\Xi$ belongs to $\cB_R(\rho)$. From $a\ell_0\le R/100$ and~\eqref{eq:curvature}, it follows that
\begin{equation}\label{eq:bsmall}
 b=(a\ell_0)^2\frac{k}{s_0^2}\le=0.0006R,
 \qquad 1\le h\le1.0006.
\end{equation}
Leting $D_k=\diag(\sqrt{\mu_1},\ldots,\sqrt{\mu_k})$, expansion of~\eqref{eq:cube} gives
\[
 \sigma_\Xi-\rho_0
 =\frac{a(DE_\Xi^\dagger+E_\Xi D)+a^2E_\Xi E_\Xi^\dagger-b\rho_0}{h}.
\]
The off-diagonal blocks in the first term are $aD_k\Xi^T$ and its adjoint, so this term has trace norm $2a\norm{D_k\Xi^T}_1$. Since row $i$ of $D_k\Xi^T$ has Euclidean norm $\sqrt{m\mu_i}$, the row bound~\eqref{eq:rownorm} implies
\begin{equation}\label{eq:detrow}
 \norm{D_k\Xi^T}_1\le\sqrt m\sum_{i=1}^k\sqrt{\mu_i}=\ell_0
 \quad\text{for every }\Xi.
\end{equation}
The remaining two matrices $a^2E_\Xi E_\Xi^\dagger$ and $b\rho_0$ are positive semidefinite and each has trace $b$. Their difference therefore has trace norm at most $2b$. Applying the triangle inequality and then~\eqref{eq:bsmall}, we obtain
\[
 T(\sigma_\Xi,\rho_0)
 \le\frac{a\ell_0+b}{h}\le0.0106R,
 \qquad
 T(\sigma_\Xi,\rho)
 \le T(\sigma_\Xi,\rho_0)+T(\rho_0,\rho)
 \le0.5106R<R.
\]
This verifies membership for every vertex. We shall later restrict a probability estimate to a typical operator-norm event, but the prior itself will remain uniform on the full cube so that every neighboring pair is available.

\subsection{Adjacent signs are difficult to distinguish}

Consider $\Xi,\Upsilon$ that differ in exactly one entry. By the normalization above, the unit vectors
\[
 \ket{\psi_\Xi}=\vecop(A_\Xi)/\sqrt h,\qquad
 \ket{\psi_\Upsilon}=\vecop(A_\Upsilon)/\sqrt h
\]
purify $\sigma_\Xi,\sigma_\Upsilon$. Since the matrices are real and differ by $2a$ in one entry, we have
\[
 4a^2=\norm{A_\Xi-A_\Upsilon}_{\HS}^2
      =2h-2\Tr A_\Xi^\dagger A_\Upsilon.
\]
Dividing by $2h$ gives
\begin{equation}\label{eq:neighboroverlap}
 \ip{\psi_\Xi}{\psi_\Upsilon}=1-2a^2/h.
\end{equation}
Since partial trace cannot increase trace distance, we may bound the distance between the $n$-copy states by the distance between their tensor-product purifications. Using~\eqref{eq:pureT} and~\eqref{eq:neighboroverlap}, we then have
\begin{align}
 T(\sigma_\Xi^{\otimes n},\sigma_\Upsilon^{\otimes n})
 &\le\sqrt{1-(1-2a^2/h)^{2n}}\notag\\
 &\le2a\sqrt{n/h}\le2a\sqrt n\le1/50,\label{eq:neighborTV}
\end{align}
where the second inequality follows from $1-(1-x)^N\le Nx$ for $0\le x\le1$, and the last two use $h\ge1$ and $a\le1/(100\sqrt n)$. From~\eqref{eq:helstrom} and~\eqref{eq:neighborTV}, it follows that every binary test between these neighboring states has equal-prior average error at least $49/100$. This conclusion holds for any collective measurement on the $n$ copies.

\subsection{Weighted Assouad decoding}

We now apply the weighted form of Assouad's hypercube argument~\cite{Yu}. Given the estimator output $\widehat\sigma$, let $\widehat\Xi$ minimize $T(\widehat\sigma,\sigma_\Upsilon)$ over the finite cube, using a fixed ordering to resolve ties. This rule defines a measurable decoder. Since the true vertex is among the candidates, its minimizing property and the triangle inequality imply
\begin{equation}\label{eq:nearest}
 T(\sigma_\Xi,\sigma_{\widehat\Xi})
 \le T(\sigma_\Xi,\widehat\sigma)
     +T(\widehat\sigma,\sigma_{\widehat\Xi})
 \le2T(\widehat\sigma,\sigma_\Xi).
\end{equation}
Let $\Xi$ be uniform on the full cube, and define the weighted number of incorrectly decoded signs by
\begin{equation}\label{eq:W}
 N_\mu(\Xi,\widehat\Xi)
 =\sum_{i=1}^k\sqrt{\mu_i}
    \sum_{j=1}^m\mathbf1\{\Xi_{ji}\ne\widehat\Xi_{ji}\},
 \qquad N_{\mu,\max}=ms_0.
\end{equation}
Fix a coordinate $(j,i)$ and fix all the other signs. Under the uniform prior, the two remaining vertices have equal probability. The decoded value $\widehat\Xi_{ji}$ is a binary test for this pair, so~\eqref{eq:neighborTV} and~\eqref{eq:helstrom} give an average error of at least $49/100$. Averaging over the other signs preserves this lower bound. Multiplying by $\sqrt{\mu_i}$ and summing over $(j,i)$, we obtain
\begin{equation}\label{eq:EW}
 \E N_\mu\ge\frac{49}{100}N_{\mu,\max},
\end{equation}
where the expectation includes the random vertex and the measurement outcome.

To obtain a probability bound, set
$\pi_*=\Prb\{N_\mu\ge N_{\mu,\max}/10\}$. Since $N_\mu/N_{\mu,\max}$ is at most $1/10$ outside this event and at most one on it, we have
\[
 \frac{\E N_\mu}{N_{\mu,\max}}
 \le\frac{1-\pi_*}{10}+\pi_*.
\]
Combining this inequality with~\eqref{eq:EW}, we have that
\begin{equation}\label{eq:Wprob}
 \Prb\{N_\mu\ge N_{\mu,\max}/10\}\ge\frac{13}{30},
\end{equation}
as desired.
The bound~\eqref{eq:Wprob} will give a lower bound on the probability of large estimation error, in addition to the expected-loss lower bound.

\subsection{The trace-norm dual witness retains all spectral scales}\label{sec:dualwitness}

We use a spectrally weighted version of the trace-distance Hamming separation of Acharya \emph{et al.}~\cite[Lemma 4.7]{ADLY}, whose dual witness is built from the true sign parameter alone. The square-root parametrization places the weights $\sqrt{\mu_i}$ in an off-diagonal block, and testing that block against the true sign matrix retains their full sum.

The top-right block of $\sigma_\Xi-\sigma_\Upsilon$ is $(a/h)D_k(\Xi-\Upsilon)^T$. Applying Lemma~\ref{lem:block} to this difference gives
\[
 T(\sigma_\Xi,\sigma_\Upsilon)
 \ge\frac a h\norm{D_k(\Xi-\Upsilon)^T}_1.
\]
Since $\Xi^T/\norm\Xi_{\op}$ has operator norm one, it is an admissible test matrix in~\eqref{eq:duality}. We therefore have
\begin{align}
 \norm{D_k(\Xi-\Upsilon)^T}_1
 &\ge\frac{\sum_{i,j}\sqrt{\mu_i}\Xi_{ji}(\Xi_{ji}-\Upsilon_{ji})}{\norm\Xi_{\op}}\notag\\
 &=\frac{2N_\mu(\Xi,\Upsilon)}{\norm\Xi_{\op}},\label{eq:weighteddual}
\end{align}
where the equality follows because $\Xi_{ji}(\Xi_{ji}-\Upsilon_{ji})$ is zero when the signs agree and is two when they disagree. In particular, all terms in the sum are nonnegative.

Taking $\Upsilon=\widehat\Xi$ and combining~\eqref{eq:nearest} with~\eqref{eq:weighteddual}, we obtain
\[
 T(\widehat\sigma,\sigma_\Xi)
 \ge\frac{aN_\mu(\Xi,\widehat\Xi)}{h\norm\Xi_{\op}}.
\]
It follows that, on the event $\norm\Xi_{\op}\le10\sqrt m$,
\begin{equation}\label{eq:lossW}
 T(\widehat\sigma,\sigma_\Xi)
 \ge\frac{aN_\mu(\Xi,\widehat\Xi)}{10h\sqrt m}.
\end{equation}
The denominator depends only on the true sign matrix, so no operator-norm bound on the decoder's output is needed, as in~\cite{ADLY}. The numerator retains the individual weights $\sqrt{\mu_i}$; this is what allows the lower bound to involve their full sum.

\subsection{An elementary bound for the sign matrix norm}

We next bound the probability of the event excluded in~\eqref{eq:lossW}. A Euclidean unit sphere in $\mathbb R^q$ has a $1/4$-net with at most $9^q$ points. Indeed, take a maximal collection of points at pairwise distance greater than $1/4$. The balls of radius $1/8$ centered at these points are disjoint and contained in the ball of radius $9/8$, so comparison of volumes bounds the collection's size by $9^q$. Its maximality gives the net property.

For fixed unit vectors $x\in\mathbb R^m$ and $y\in\mathbb R^k$, independence of the signs gives
\[
 \E e^{u x^T\Xi y}
 =\prod_{j,i}\cosh(ux_jy_i)
 \le\exp\left(\frac{u^2}{2}\sum_{j,i}x_j^2y_i^2\right)
 =e^{u^2/2},
\]
where we used $\cosh z\le e^{z^2/2}$ and $\sum_{j,i}x_j^2y_i^2=1$. Applying Markov's inequality with $u=t$ bounds the upper tail by $e^{-t^2/2}$. Applying the same argument to $-x^T\Xi y$ therefore yields
\[
 \Prb\{|x^T\Xi y|\ge t\}\le2e^{-t^2/2}.
\]
Choose unit vectors $x,y$ attaining $|x^T\Xi y|=\norm\Xi_{\op}$, and choose net points $x_0,y_0$ within distance $1/4$ of them. Since
\[
 |x^T\Xi y-x_0^T\Xi y_0|
 \le\norm\Xi_{\op}\bigl(\norm{x-x_0}_2+\norm{y-y_0}_2\bigr)
 \le\tfrac12\norm\Xi_{\op},
\]
the event $\norm\Xi_{\op}>10\sqrt m$ implies that some pair of net points satisfies $|x_0^T\Xi y_0|>5\sqrt m$. A union bound over at most $9^{m+k}$ pairs then gives
\begin{align}
 \Prb\{\norm\Xi_{\op}>10\sqrt m\}
 &\le2\,9^{m+k}e^{-25m/2}<10^{-3}.\label{eq:signnorm}
\end{align}

\subsection{Complete the lower bound}

From~\eqref{eq:Wprob} and~\eqref{eq:signnorm}, the union bound implies that
\[
 \Prb\{N_\mu\ge N_{\mu,\max}/10,
          \ \norm\Xi_{\op}\le10\sqrt m\}
 \ge13/30-10^{-3}>0.43.
\]
No independence between these two events is required. On their intersection, $N_{\mu,\max}=ms_0$, $\ell_0=\sqrt m\,s_0$, and~\eqref{eq:lossW} give
\[
 T(\widehat\sigma,\sigma_\Xi)
 \ge\frac{ams_0}{100h\sqrt m}
 =\frac{a\ell_0}{100h}
 \ge\frac{a\ell_0}{200},
\]
where the last inequality uses $h\le1.0006<2$. Moreover, from~\eqref{eq:cube} and~\eqref{eq:Llower}, it follows that
\[
 a\ell_0=\frac1{100}\min\{\ell_0/\sqrt n,R\}
 \ge\frac1{900}\min\{\sqrt{\Gamma/n},R\}
 =\frac{q_n}{900}.
\]
Consequently, the prior-averaged probability of loss at least $q_n/180000$ is greater than $0.43$. Since this average is over finitely many vertices, at least one vertex has a probability greater than $0.43$. Because $10^{-6}<1/180000$, the same vertex satisfies the strict-tail bound~\eqref{eq:stronglower}.

The same event also gives the expected-loss bound. Since every vertex belongs to $\cB_R(\rho)$, averaging first over the uniform prior yields
\[
 \sup_{\sigma\in\cB_R(\rho)}\E_\sigma T(\widehat\sigma,\sigma)
 \ge\E T(\widehat\sigma,\sigma_\Xi)
 \ge\frac{0.43}{180000}q_n>10^{-6}q_n.
\]
As the measurement and estimator were arbitrary, taking the infimum over them proves the expected minimax lower bound.

Finally, every procedure has worst-case failure probability greater than $1/3$ at threshold $10^{-6}q_n$, by~\eqref{eq:stronglower}. Since decreasing the threshold cannot decrease this failure probability, no threshold below $10^{-6}q_n$ is feasible in the definition of the confidence radius. Taking the infimum over feasible thresholds and then over procedures gives
$\cQ_n(\rho,R,1/3)\ge10^{-6}q_n$. This reasoning does not require either infimum to be attained.

\begin{remark}[Relation to known rank-constrained lower bounds]\label{rem:rankvariant}
The full-dimensional specialization recovers the known $\Omega(d^2/\varepsilon^2)$ bound of Haah \emph{et al.}~\cite{HHJWY}. Applied to a normalized projector without regularization, the same construction preserves rank; Appendix~\ref{app:rank} records this variant. Its local lower bound is also implied by the proof of~\cite[Proposition 4.13]{NZ}, and its global consequence is the known $\Omega(rd/\varepsilon^2)$ bound~\cite{SSW,NZ}. Unlike their projector constructions, our rank-preserving states generally have unequal nonzero eigenvalues.
\end{remark}

\section{Upper bounds from purification-based tomography}\label{sec:upper}

We now prove the upper bounds by analyzing the existing purification/Hayashi procedure. The measurement uses neither the center nor its spectrum. The rates obtained are implied by earlier moment and concentration bounds for this procedure (Remark~\ref{rem:priorupper}); the direct argument below is short and gives explicit constants. We first derive a bound for each input state, and then use the known center and radius to make this bound uniform over the ball.

\subsection{The random-purification channel}

The random-purification channel of Tang, Wright, and Zhandry~\cite[Lemma 2.11]{TWZ}, restated as Theorem 1.1 of Pelecanos, Spilecki, Tang, and Wright~\cite{PSTW}, gives a quantum channel $\mathcal P_{n,d}$ satisfying
\begin{equation}\label{eq:purify}
 \mathcal P_{n,d}(\sigma^{\otimes n})
  =\E_A\left[\bigl(\ket{\vecop(A)}\bra{\vecop(A)}\bigr)^{\otimes n}\right],
 \qquad AA^\dagger=\sigma,
\end{equation}
where the expectation is over the unitarily invariant distribution of purifications in $\mathbb C^d\otimes\mathbb C^d$. The channel depends only on $n$ and $d$. All $n$ pure states on the right side share the same random purification, and the estimator is not told which purification was chosen.

The original lemma allows an ancillary dimension $q$ whenever $\operatorname{rank}\sigma\le q$. We take $q=d$ so that the same procedure applies to every state, including singular states, without knowing its rank. Girardi, Mele, and Lami~\cite{GML} subsequently gave a simplified construction and extensions of this channel. Equation~\eqref{eq:purify} produces a mixture of tensor powers, and does not assert that one can deterministically purify an unknown state. In the analysis below we condition on $A$ and obtain bounds that hold for every $A$; averaging then gives the same bounds for the channel output.

\subsection{The estimator and the exact outcome law}

Let $D_*=d^2$ be the dimension of the purification space. After applying~\eqref{eq:purify}, we use Hayashi's covariant POVM~\cite{Hayashi}
\begin{equation}\label{eq:hayashi}
 \left\{\binom{n+D_*-1}{n}\bigl(\ket v\bra v\bigr)^{\otimes n}\,dv\right\},
\end{equation}
where $dv$ denotes normalized Haar measure on pure states. Let $\Pi_{\mathrm{sym}}$ be the projector onto the symmetric subspace of $(\mathbb C^{D_*})^{\otimes n}$. The Haar moment identity~\cite[Equation (9)]{PSTW}
\[
 \int\bigl(\ket v\bra v\bigr)^{\otimes n}dv
 =\frac{\Pi_{\mathrm{sym}}}{\binom{n+D_*-1}{n}}
\]
shows that the operators in~\eqref{eq:hayashi} integrate to $\Pi_{\mathrm{sym}}$. The \textcolor{red}{output} in~\eqref{eq:purify} is supported on the symmetric subspace, so adding the single effect $I-\Pi_{\mathrm{sym}}$ completes the POVM without changing the outcome law.

Given the outcome $v=\vecop(V)$, define
\begin{equation}\label{eq:estimate}
 \widehat\sigma=VV^\dagger.
\end{equation}
Since $\norm V_{\HS}=1$, the output is a density matrix. This is the existing $\mathrm{Mix}(\mathrm{Hayashi})$ procedure with ancillary dimension $d$.

We next describe its outcome law. Condition on a purification $A$ in~\eqref{eq:purify} and put $a_0=\vecop(A)$. By~\eqref{eq:hayashi}, the outcome has density proportional to $|\ip{a_0}{v}|^{2n}$ with respect to $dv$. Its squared overlap with $a_0$ therefore has the $\operatorname{Beta}(n+1,D_*-1)$ distribution, as in the proof of Proposition 5.1 in~\cite[Section 5.1]{PSTW}, while its normalized component orthogonal to $a_0$ remains Haar distributed and independent of the overlap. We use an equivalent Gaussian form of this law. Let $X$ and $G$ be independent random variables with
\begin{equation}\label{eq:gaussianlaw}
 X\sim\GammaDist(n+1,1),\qquad
 \vecop(G)\sim\mathrm{CN}(0,I-a_0a_0^\dagger),\qquad Q=\norm G_{\HS}^2,
\end{equation}
where, for the positive integer shapes used here, $\GammaDist(\alpha,1)$ has density $x^{\alpha-1}e^{-x}/(\alpha-1)!$ on $x>0$. A standard complex Gaussian coordinate has independent real and imaginary parts with law $N(0,1/2)$. Thus $G$ is a standard complex Gaussian matrix on the orthogonal complement of $A$, so that
\[
 Q\sim\GammaDist(D_*-1,1),\qquad \Tr A^\dagger G=0,
\]
and $G/\sqrt Q$ is Haar distributed on the unit sphere of that complement and independent of $Q$. Since $X/(X+Q)$ then has the $\operatorname{Beta}(n+1,D_*-1)$ distribution, the outcome matrix has, up to a global phase, the exact distribution
\begin{equation}\label{eq:exactV}
 V=\frac{\sqrt X A+G}{\sqrt{X+Q}}.
\end{equation}
The formula is exact for finite $n$ and does not use an asymptotic approximation.

\subsection{Separate the linear and quadratic errors}

Since $AA^\dagger=\sigma$, expanding~\eqref{eq:exactV} gives
\begin{equation}\label{eq:expandV}
 \widehat\sigma-\sigma
 =\frac{\sqrt X(AG^\dagger+GA^\dagger)+GG^\dagger-Q\sigma}{X+Q}.
\end{equation}
Set $f(G)=\norm{AG^\dagger+GA^\dagger}_1$. Both $GG^\dagger$ and $Q\sigma$ are positive with trace $Q$, so their difference has trace norm at most $2Q$. Applying the triangle inequality to~\eqref{eq:expandV}, and then using $X+Q\ge X$, we obtain
\begin{equation}\label{eq:basicupper}
 T(\widehat\sigma,\sigma)
 \le\frac{\sqrt X f(G)+2Q}{2(X+Q)}
 \le\frac{f(G)}{2\sqrt X}+\frac QX.
\end{equation}
We bound the linear term through the spectrum of $\sigma$ and the quadratic term through the scalar $Q$.

Let $A=\sum_i\sqrt{\nu_i}u_iv_i^\dagger$ be a singular value decomposition, where the $\nu_i$ are the eigenvalues of $\sigma$. Since a rank-one matrix $xy^\dagger$ has trace norm $\norm x_2\norm y_2$, the triangle inequality gives
\[
 AG^\dagger=\sum_i\sqrt{\nu_i}\,u_i(Gv_i)^\dagger,
 \qquad \norm{AG^\dagger}_1\le\sum_i\sqrt{\nu_i}\norm{Gv_i}_2.
\]
The covariance in~\eqref{eq:gaussianlaw} is bounded above by the identity. Since the map $G\mapsto Gv_i$ has squared Hilbert--Schmidt norm $d$, it follows that $\E\norm{Gv_i}_2^2\le d$. Therefore $\E\norm{Gv_i}_2\le\sqrt d$ by Cauchy--Schwarz. Combining these bounds with $\norm{GA^\dagger}_1=\norm{AG^\dagger}_1$, we obtain
\begin{equation}\label{eq:meanf}
 \E f(G)\le2\sqrt d\sum_i\sqrt{\nu_i}=2\sqrt d\,S(\sigma).
\end{equation}
The right side depends only on $AA^\dagger$, so the bound holds uniformly in the purification $A$.

\subsection{Expected risk of the universal estimator}

\begin{proposition}[Statewise expected bound]\label{prop:stateupper}
For every $\sigma\in\cD_d$ and every $n\ge1$, the proper estimator~\eqref{eq:estimate} satisfies
\begin{equation}\label{eq:stateexpected}
 \E_\sigma T(\widehat\sigma,\sigma)
 \le\sqrt{\frac dn}\,S(\sigma)+\frac{d^2-1}{n}.
\end{equation}
\end{proposition}
\begin{proof}
Since $X\sim\GammaDist(n+1,1)$, integration of its density gives $\E X^{-1}=1/n$, and Cauchy--Schwarz then gives $\E X^{-1/2}\le n^{-1/2}$. Also, $Q\sim\GammaDist(d^2-1,1)$ implies $\E Q=d^2-1$. As $X$ is independent of $G$ and $Q$, it follows from~\eqref{eq:basicupper} and~\eqref{eq:meanf} that
\[
 \E\,T(\widehat\sigma,\sigma)
 \le\tfrac12(\E f(G))(\E X^{-1/2})+(\E Q)(\E X^{-1})
 \le\sqrt{\frac dn}\,S(\sigma)+\frac{d^2-1}{n}.
\]
This bound holds conditionally on every $A$. Averaging over the purifications in~\eqref{eq:purify} proves~\eqref{eq:stateexpected}.
\end{proof}

The next proposition records the exact bias of~\eqref{eq:estimate}, a direct consequence of the known unbiasedness of the debiased outcome~\cite[Lemma 5.7]{PSTW}, and the resulting lower bound on its risk.

\begin{proposition}[Exact bias and scope of the statewise bound]\label{prop:bias}
For the estimator~\eqref{eq:estimate},
\begin{equation}\label{eq:exactbias}
 \E_\sigma\widehat\sigma-\sigma
 =\frac{d(I-d\sigma)}{n+d^2}.
\end{equation}
In particular, if $\sigma=P_r/r$ for a rank-$r$ projector, then
\begin{equation}\label{eq:biaslower}
 \E_\sigma T(\widehat\sigma,\sigma)
 \ge\frac{d(d-r)}{n+d^2}.
\end{equation}
\end{proposition}
\begin{proof}
Condition on $A$ in~\eqref{eq:purify}. For Hayashi's measurement on $n$ copies of the pure state $\ket{a_0}$ in dimension $D_*$, the debiased outcome $\frac{n+D_*}{n}\ket v\bra v-\frac1nI_{D_*}$ is an unbiased estimator of $\ket{a_0}\bra{a_0}$~\cite[Lemma 5.7 and Corollary 5.9]{PSTW}. Taking the partial trace over the $d$-dimensional ancilla gives $\frac{n+d^2}{n}\E\widehat\sigma-\frac dnI=\sigma$, that is,
\[
 \E\widehat\sigma=\frac{n\sigma+dI}{n+d^2}.
\]
The last expression is independent of $A$, so averaging proves~\eqref{eq:exactbias}. By convexity of the trace norm,
\[
 \E_\sigma T(\widehat\sigma,\sigma)
 \ge\tfrac12\norm{\E_\sigma\widehat\sigma-\sigma}_1
 =\frac{d}{2(n+d^2)}\norm{I-d\sigma}_1.
\]
When $\sigma=P_r/r$, the eigenvalues of $I-d\sigma$ are $1-d/r$ on the range of $P_r$ and $1$ on its orthogonal complement. Thus $\norm{I-d\sigma}_1=2(d-r)$, which gives~\eqref{eq:biaslower}.
\end{proof}

\begin{remark}[The role of the ancillary dimension]\label{rem:ancilladimension}
Suppose a rank bound is known and use an ancillary dimension $q$ with $\operatorname{rank}\sigma\le q\le d$. The same construction uses $d\times q$ matrices $A,G$, with $AA^\dagger=\sigma$, purification dimension $dq$, and $Q\sim\GammaDist(dq-1,1)$. The unbiasedness used in the proof of Proposition~\ref{prop:bias}, now in dimension $dq$ with a $q$-dimensional ancilla, yields
\begin{equation}\label{eq:rectangularbias}
 \E_\sigma\widehat\sigma_q=\frac{n\sigma+qI}{n+dq},
 \qquad
 \E_\sigma\widehat\sigma_q-\sigma=\frac{q(I-d\sigma)}{n+dq}.
\end{equation}
The linear estimate~\eqref{eq:meanf} is unchanged: for each right singular vector $v_i\in\mathbb C^q$, the vector $Gv_i$ still has $d$ coordinates and $\E\norm{Gv_i}_2^2\le d$. Replacing $\E Q=d^2-1$ by $dq-1$ in the proof of Proposition~\ref{prop:stateupper} consequently gives
\begin{equation}\label{eq:knownrankexpected}
 \E_\sigma T(\widehat\sigma_q,\sigma)
 \le\sqrt{\frac dn}\,S(\sigma)+\frac{dq-1}{n}.
\end{equation}
In particular, for a known rank bound $r$, take $q=r$ and use $S(\sigma)\le\sqrt r$ to obtain
\[
 \E_\sigma T(\widehat\sigma_r,\sigma)
 \le\sqrt{\frac{dr}{n}}+\frac{dr}{n}.
\]
At $\sigma=P_r/r$, the bias in trace distance equals $r(d-r)/(n+dr)$; when $n\ge dr$, this is at most $dr/n\le\sqrt{dr/n}$. Thus the $d^2/n$ remainder concerns the choice $q=d$ in the particular estimator~\eqref{eq:estimate}. It is not a limitation of purification/Hayashi tomography with a suitable rank bound, nor does it rule out a better procedure when the rank is unknown.
\end{remark}

Proposition~\ref{prop:bias} shows why the remainder cannot simply be dropped from a uniform bound for~\eqref{eq:estimate}. At a pure state and $n=d^2$, its expected error is at least $(d-1)/(2d)$, whereas~\eqref{eq:knownrankexpected} with $q=1$ gives expected error of order $d^{-1/2}$. For a general spectrum, $n\ge d^3/S(\sigma)^2$ implies $d^2/n\le\sqrt{d/n}\,S(\sigma)$, so this is a sufficient regime in which the displayed spectral term dominates the remainder. It is not asserted to be a necessary threshold. For the local problem, we will return the center at small sample sizes and use the radius term to absorb $d^2/n$ in the remaining regime.

\subsection{Exponential tails with a dimension-independent confidence term}

We first control $f(G)$. By the reverse triangle inequality and~\eqref{eq:holder}, for any two matrices $G,G'$ we have
\begin{align}
 |f(G)-f(G')|
 &\le\norm{A(G-G')^\dagger+(G-G')A^\dagger}_1\notag\\
 &\le2\norm A_{\HS}\norm{G-G'}_{\HS}
 =2\norm{G-G'}_{\HS}.\label{eq:lipschitz}
\end{align}
The last equality uses $\norm A_{\HS}=1$. Since the map from real standard-Gaussian coordinates to $G$ has operator norm $1/\sqrt2$,~\eqref{eq:lipschitz} shows that $f$ is $\sqrt2$-Lipschitz in those coordinates. Applying the Gaussian concentration bound~\eqref{eq:gaussconc} from Appendix~\ref{app:gaussian}, and then using~\eqref{eq:meanf}, gives
\begin{equation}\label{eq:fconcentration}
 \Prb\{f(G)>2\sqrt d\,S(\sigma)+2\sqrt t\}\le e^{-t}
 \quad(t>0).
\end{equation}
In particular, the fluctuation increment $2\sqrt t$ has no dimension factor.

We next control $Q$ and $X$. Since $\E e^{Q/2}=2^{D_*-1}$, exponential Markov inequality gives
\begin{equation}\label{eq:qtail}
 \Prb\{Q>2(D_*-1)+2t\}
 \le e^{-(D_*-1)-t}2^{D_*-1}\le e^{-t},
\end{equation}
where the last inequality follows from $\log2<1$. Likewise, $\E e^{-X}=2^{-(n+1)}$ gives
\begin{equation}\label{eq:xtail}
 \Prb\{X<n/2\}\le e^{n/2}2^{-(n+1)}
 =\tfrac12 e^{-(\log2-1/2)n}.
\end{equation}
As $6(\log2-1/2)>1$, the last expression is at most $e^{-t}$ whenever $n\ge6t$.

\begin{proposition}[Statewise high-probability bound]\label{prop:tail}
For every $\sigma\in\cD_d$, every $n\ge1$, and every $t>0$, with probability at least $1-3e^{-t}$,
\begin{equation}\label{eq:statetail}
 T(\widehat\sigma,\sigma)
 \le\sqrt{\frac2n}\bigl(\sqrt d\,S(\sigma)+\sqrt t\bigr)
       +\frac{4(d^2+t)}{n}.
\end{equation}
\end{proposition}
\begin{proof}
Suppose first that $n\ge6t$. By~\eqref{eq:fconcentration}--\eqref{eq:xtail} and the union bound, with probability at least $1-3e^{-t}$ we have
\[
 f(G)\le2\sqrt d\,S(\sigma)+2\sqrt t,\qquad
 Q\le2(d^2-1)+2t,\qquad X\ge n/2.
\]
Substituting these three inequalities into~\eqref{eq:basicupper} proves~\eqref{eq:statetail}. If $n<6t$, the two terms involving $t$ on its right side already exceed $\sqrt{1/3}+2/3>1$. Since the trace distance between states is at most one,~\eqref{eq:statetail} then holds deterministically. The argument is uniform in $A$, so averaging over~\eqref{eq:purify} completes the proof.
\end{proof}
The probability bound is nonvacuous when $t>\log3$; no restriction on $t$ is needed for the statement.

\subsection{Make the estimator uniform over the known ball}\label{sec:localize}

We now compare $S(\sigma)$ with the spectrum of the known center. The Lidskii--Mirsky inequality~\cite{Bhatia} states that the ordered eigenvalues of two Hermitian matrices satisfy
\[
 \sum_i|\lambda_i(\sigma)-\lambda_i(\rho)|\le\norm{\sigma-\rho}_1.
\]
It applies without a commutativity assumption; Appendix~\ref{app:eigen} also gives a proof through eigenvalue variation. Since $|\sqrt x-\sqrt y|\le\sqrt{|x-y|}$ for nonnegative $x,y$, for every $\sigma\in\cB_R(\rho)$ we obtain
\begin{equation}\label{eq:Sloc}
 S(\sigma)\le S(\rho)+\sum_i\sqrt{|\lambda_i(\sigma)-\lambda_i(\rho)|}
 \le S(\rho)+\sqrt{2dR},
\end{equation}
where the second inequality follows from Cauchy--Schwarz, the eigenvalue inequality, and $\norm{\sigma-\rho}_1\le2R$. Multiplying~\eqref{eq:Sloc} by $\sqrt d$ and applying Cauchy--Schwarz to the two resulting terms gives
\begin{equation}\label{eq:sqrtGamma}
 \sqrt d\,S(\sigma)\le\sqrt{dS(\rho)^2}+\sqrt{2d^2R}\le\sqrt{3\Gamma}.
\end{equation}

For the expected-risk upper bound in Theorem~\ref{thm:main}, return $\rho$ when $n\le\Gamma/R^2$, and use~\eqref{eq:estimate} otherwise. In the first regime, returning $\rho$ gives error at most $R=q_n$. In the second regime, $n>\Gamma/R^2\ge d^2/R$ and $q_n=\sqrt{\Gamma/n}$. Since $n>d^2/R$ and $d^2R\le\Gamma$, we then have
\begin{equation}\label{eq:absorb}
 \frac{d^2}{n}\le d\sqrt{\frac Rn}\le\sqrt{\frac\Gamma n}.
\end{equation}
From~\eqref{eq:stateexpected},~\eqref{eq:sqrtGamma}, and~\eqref{eq:absorb}, it follows that the expected error is at most $(\sqrt3+1)\sqrt{\Gamma/n}$, which proves the expected-risk upper bound.

For the confidence-radius upper bound, fix $0<\delta\le1/3$ and put
\[
 \ell=\log(1/\delta),\qquad \Gamma_\delta=\Gamma+\ell,
 \qquad t=\log(3/\delta).
\]
We return $\rho$ when $n\le \Gamma_\delta/R^2$, and use~\eqref{eq:estimate} otherwise. The rule depends only on the known center, the radius, the sample size, and the requested confidence. In the first regime, its error is at most $R=\min\{R,\sqrt{\Gamma_\delta/n}\}$. In the second regime, Proposition~\ref{prop:tail} and~\eqref{eq:sqrtGamma} give, uniformly over the ball,
\begin{equation}\label{eq:balltail}
 T(\widehat\sigma,\sigma)\le\sqrt{\frac{6\Gamma}{n}}+
     \sqrt{\frac{2t}{n}}+\frac{4(d^2+t)}n
\end{equation}
with failure probability at most $\delta$. Since $\ell\ge\log3$, we have $t\le2\ell$, and Cauchy--Schwarz bounds the two square-root terms in~\eqref{eq:balltail} by $\sqrt{10\Gamma_\delta/n}$. Also, $n>\Gamma_\delta/R^2\ge\Gamma/R^2$ allows us to apply~\eqref{eq:absorb}, whereas $\Gamma_\delta/n<R^2\le1$ gives
\[
 \frac{d^2}{n}\le\sqrt{\frac{\Gamma_\delta}{n}},
 \qquad \frac tn\le\frac{2\Gamma_\delta}{n}\le2\sqrt{\frac{\Gamma_\delta}{n}}.
\]
Substituting these bounds into~\eqref{eq:balltail}, we obtain
\begin{equation}\label{eq:fullconfidenceupper}
 \cQ_n(\rho,R,\delta)
 \le16\min\left\{R,\sqrt{\frac{\Gamma+\log(1/\delta)}n}\right\}.
\end{equation}
 Thus the same statewise measurement proves both upper bounds; only the decision to return the center depends on the loss criterion and confidence level.

\begin{remark}[The upper rates follow from earlier results]\label{rem:priorupper}
The estimator $\mathrm{Mix}(\mathcal A_{\mathrm{GPS}})$ of~\cite{PSTW,PSTWspectrum} with a $d$-dimensional ancilla uses the measurement of this section and outputs the unbiased Hermitian matrix $\widetilde\sigma=\frac{n+d^2}{n}\widehat\sigma-\frac dnI$, which need not be positive. The results cited below are stated for rank-$r$ states and an $r$-dimensional ancilla; we use $r=d$, which is permitted for every $\sigma$ because the purification channel only requires the ancilla dimension to be at least $\operatorname{rank}\sigma$. Put $Z=\widetilde\sigma-\sigma$. Since $\widehat\sigma-\sigma=\frac{n}{n+d^2}Z+\frac{d(I-d\sigma)}{n+d^2}$ and $\norm{I-d\sigma}_1\le2d$, we have $T(\widehat\sigma,\sigma)\le\frac12\norm Z_1+d^2/n$. Existing bounds on $Z$ then give~\eqref{eq:priorupper} and, more precisely, the rates of Propositions~\ref{prop:stateupper} and~\ref{prop:tail}, although none of the three statements below appears in the cited papers. In an eigenbasis of $\sigma$ with eigenvalues $\nu_i$, let
\[
 B=\sqrt{n/d}\sqrt{\sigma+(d/n)I},\qquad
 K_{ij}=\frac{2Z_{ij}}{B_{ii}+B_{jj}},
\]
so that $Z=(BK+KB)/2$ and $\Tr B\le\sqrt{n/d}\,S(\sigma)+d$.

\emph{Expectation.} The second-moment formula~\cite[Theorem 1.7]{PSTW} gives $\E|Z_{ij}|^2\le(\nu_i+\nu_j)/n+d/n^2$: after $\sigma\otimes\sigma$ is subtracted, the terms $-\sigma\otimes\sigma/n$ and the sum-of-squares remainder contribute nonpositively to these entries. Hence $\E|K_{ij}|^2\le4d/n^2$. Since $K$ is Hermitian, $\norm{KB}_1=\norm{BK}_1$, and the row bound~\eqref{eq:rownorm} with Cauchy--Schwarz yields
\[
 \E\norm Z_1\le\E\norm{BK}_1\le\frac{2d}n\Tr B
 \le2\sqrt{d/n}\,S(\sigma)+2d^2/n.
\]

\emph{Constant confidence.} The final unnumbered estimate following Equation (31) in the proof of Theorem 4.8 of~\cite{PSTWspectrum} states, in this notation, that $\norm K_{\op}\le432d/n$ with probability at least $0.99$. On this event, $\norm Z_1\le\norm K_{\op}\Tr B\le432\sqrt{d/n}\,S(\sigma)+432d^2/n$.

\emph{All confidence levels.} Applying~\cite[Proposition 4.7 and Lemma 3.6]{PSTWspectrum} to normalized differences shows that, for Hermitian $H,H'$ of operator norm at most one and $u>0$, the increment $\Tr((H-H')Z)$ exceeds $\sqrt u\,d_2(H,H')+u\,d_1(H,H')$ in absolute value with probability at most $2e^{-u}$, where
\[
 d_2(H,H')^2=\frac{16\Tr\bigl(\sigma(H-H')^2\bigr)}n+\frac{8d\Tr(H-H')^2}{n^2},
 \qquad d_1(H,H')=\frac{2\norm{H-H'}_{\op}}n.
\]
Since $\norm Z_1$ is the supremum of $\Tr(HZ)$ over such $H$, generic chaining~\cite[Theorem 3.5]{Dirksen} applies. Its $\gamma_2$ functional is bounded, through the majorizing-measure theorem~\cite[Remark 3.3(i)]{Dirksen} and~\eqref{eq:rownorm}, by the expected trace norm of a Gaussian matrix with canonical metric $d_2$, and its $\gamma_1$ functional by an entropy integral. This gives, with probability at least $1-e^{-u}$ for $u\ge1$,
\[
 \norm Z_1\le C\left[\sqrt{d/n}\,S(\sigma)+\sqrt{u/n}+(d^2+u)/n\right].
\]
We only outline this deduction, since Proposition~\ref{prop:tail} proves the rate directly with explicit constants.
\end{remark}

\section{Confidence lower bound and copy complexity}\label{sec:confidenceLower}

Section~\ref{sec:lower} gives the dimension and spectrum terms in the converse. We now prove the additional confidence term and combine the two bounds to complete Theorem~\ref{thm:main}. We construct two states in the same ball and apply Le Cam's two-point method~\cite{Yu}, using a diagonal perturbation when no eigenvalue dominates and a coherent rotation otherwise. We then obtain Corollary~\ref{thm:confidence} by inverting the confidence-radius bounds.

\begin{lemma}[Two local alternatives]\label{lem:twopoint}
For every $d\ge2$, $\rho\in\cD_d$, and $0<\varepsilon\le\min\{R/4,1/64\}$, there are $\sigma_+,\sigma_-\in\cB_R(\rho)$ with
\[
 T(\sigma_+,\sigma_-)>2\varepsilon,
 \qquad -\log F(\sigma_+,\sigma_-)\le128\varepsilon^2.
\]
\end{lemma}
\begin{proof}
Write $\rho$ in an eigenbasis with eigenvalues $\lambda_1\ge\cdots\ge\lambda_d$. We first construct the two states and then bound their fidelity.

\emph{Case 1: $\lambda_1\le3/4$.} We claim that some subset $I$ of the coordinates has total mass $\beta\in[1/4,3/4]$. If one eigenvalue is at least $1/4$, we take that coordinate. Otherwise, we add coordinates until their total first reaches $1/4$; since each added eigenvalue is less than $1/4$, the total is less than $1/2$. This proves the claim.

Let $u=2\varepsilon$ and define diagonal matrices by
\[
 (\sigma_\pm)_{ii}=\begin{cases}
   (\beta\pm u)\lambda_i/\beta,&i\in I,\\
   (1-\beta\mp u)\lambda_i/(1-\beta),&i\notin I.
 \end{cases}
\]
Since $u\le1/32$ and $\beta,1-\beta\ge1/4$, these matrices are positive semidefinite. Their traces equal $(\beta\pm u)+(1-\beta\mp u)=1$, so they are states. Summing the absolute diagonal differences, we then have
\begin{equation}\label{eq:diagonalseparation}
 T(\sigma_\pm,\rho)=u\le R,
 \qquad T(\sigma_+,\sigma_-)=2u=4\varepsilon.
\end{equation}
Since the states are diagonal in the same basis, their fidelity satisfies
\begin{align}
 \sqrt{F(\sigma_+,\sigma_-)}
 &=\sqrt{\beta^2-u^2}+\sqrt{(1-\beta)^2-u^2}\notag\\
 &\ge1-u^2\left(\frac1\beta+\frac1{1-\beta}\right)
 \ge1-\frac{16}3u^2.\label{eq:diagonalfidelity}
\end{align}
Here the first inequality follows from $\sqrt{1-x}\ge1-x$ for $0\le x\le1$, and the second follows from $\beta(1-\beta)\ge3/16$. The right side of~\eqref{eq:diagonalfidelity} is positive because $u\le1/32$. Squaring and using $(1-x)^2\ge1-2x$, we obtain
\[
 F(\sigma_+,\sigma_-)\ge1-\frac{32}3u^2
 \ge1-64\varepsilon^2.
\]

\emph{Case 2: $\lambda_1>3/4$.} Since $d\ge2$ and the eigenvalues sum to one, we may choose a second eigenvector with $\lambda_2\le1/4$. Put $g=\lambda_1-\lambda_2>1/2$, and define
\[
 U_\theta=\begin{pmatrix}\cos\theta&-\sin\theta\\\sin\theta&\cos\theta\end{pmatrix},
 \qquad \theta=4\varepsilon,
\]
on these two coordinates, with $U_\theta$ equal to the identity elsewhere. Let $\sigma_\pm=U_{\pm\theta}\rho U_{\pm\theta}^\dagger$. These are states by unitary conjugation. The nonzero block of $\sigma_\pm-\rho$ has eigenvalues $\pm g\sin\theta$, while the difference $\sigma_+-\sigma_-$ has eigenvalues $\pm g\sin(2\theta)$. Therefore
\begin{equation}\label{eq:rotationseparation}
 T(\sigma_\pm,\rho)=g\sin\theta\le4\varepsilon\le R,
 \qquad T(\sigma_+,\sigma_-)=g\sin(2\theta).
\end{equation}
Since $0<2\theta=8\varepsilon\le1/8$ and $\sin x\ge2x/\pi$ on $[0,\pi/2]$, it follows from~\eqref{eq:rotationseparation} and $g>1/2$ that
$T(\sigma_+,\sigma_-)>(8/\pi)\varepsilon>2\varepsilon$.

To bound the fidelity, apply $U_{\pm\theta}$ to the first register of the canonical purification $\sum_i\sqrt{\lambda_i}\ket i\ket i$. Since their overlap equals $\Tr(\rho U_{2\theta})$, we have
\begin{align}
 \Tr(\rho U_{2\theta})
 &=1-(\lambda_1+\lambda_2)(1-\cos2\theta)\notag\\
 &\ge1-2\theta^2=1-32\varepsilon^2,\label{eq:rotationoverlap}
\end{align}
where the inequality follows from $\lambda_1+\lambda_2\le1$ and $1-\cos x\le x^2/2$. Fidelity is at least the squared overlap of any pair of purifications. Hence, by~\eqref{eq:rotationoverlap} and $(1-x)^2\ge1-2x$,
$F(\sigma_+,\sigma_-)\ge1-64\varepsilon^2$.

We have now proved the required separation and the bound $F\ge1-64\varepsilon^2$ in both cases. Since $64\varepsilon^2\le1/64<1/2$, applying $-\log(1-x)\le2x$ gives
$-\log F\le128\varepsilon^2$, which completes the proof.
\end{proof}

Suppose a procedure estimates both states in Lemma~\ref{lem:twopoint} with error at most $\varepsilon$ and failure probability at most $\delta$. Decode the state closer to the estimate. Since the two states are more than $2\varepsilon$ apart, the triangle inequality shows that every successful estimate gives the correct decision. The resulting equal-prior testing error is therefore at most $\delta$. From~\eqref{eq:testingF}, it follows that
\[
 \frac{1-\sqrt{1-F(\sigma_+,\sigma_-)^n}}2\le\delta.
\]
Since $1-2\delta>0$, rearranging and squaring gives $F(\sigma_+,\sigma_-)^n\le4\delta(1-\delta)$. Taking logarithms and applying Lemma~\ref{lem:twopoint}, we then obtain
\begin{equation}\label{eq:binarycopylower}
 n\ge\frac{\log(1/[4\delta(1-\delta)])}
                 {-\log F(\sigma_+,\sigma_-)}
 \ge\frac{\log(1/[4\delta(1-\delta)])}{128\varepsilon^2}.
\end{equation}
The same testing argument excludes $n=0$.

We next compare the numerator with $\log(1/\delta)$. If $\delta\le1/8$, then
\[
 \log\frac1{4\delta(1-\delta)}
 \ge\log\frac1{4\delta}
 \ge\frac13\log\frac1\delta.
\]
If $1/8\le\delta\le1/3$, then $4\delta(1-\delta)\le8/9$ and $\log(1/\delta)\le\log8$. Since $\log(9/8)>\log8/20$, the two cases give
\begin{equation}\label{eq:logcomparison}
 \log\frac1{4\delta(1-\delta)}\ge\frac1{20}\log\frac1\delta.
\end{equation}
Combining~\eqref{eq:binarycopylower} and~\eqref{eq:logcomparison}, we conclude that every successful procedure satisfies
\begin{equation}\label{eq:loglower}
 n\ge\frac{\log(1/\delta)}{2560\varepsilon^2}
\end{equation}
in the accuracy range of Lemma~\ref{lem:twopoint}.

\subsection{The confidence-radius lower bound}

Fix $n\ge1$ and $0<\delta\le1/3$, and put $\ell=\log(1/\delta)$. We first show that
\begin{equation}\label{eq:confidenceonlylower}
 \cQ_n(\rho,R,\delta)\ge\frac1{100}
       \min\{R,\sqrt{\ell/n}\}.
\end{equation}
Let $\varepsilon_*=(1/100)\min\{R,\sqrt{\ell/n}\}$. Since $R\le1$, we have $\varepsilon_*\le R/4$ and $\varepsilon_*\le1/100<1/64$, so Lemma~\ref{lem:twopoint} applies at this accuracy. If an $n$-copy procedure had failure probability at most $\delta$ throughout the ball at threshold $\varepsilon_*$, then~\eqref{eq:loglower} would give
\[
 n\ge\frac{\ell}{2560\varepsilon_*^2}
 \ge\frac{10000}{2560}n>n,
\]
where the second inequality uses $\varepsilon_*^2\le\ell/(10000n)$. This contradiction shows that $\varepsilon_*$, and every smaller threshold, is infeasible for every procedure. Taking the infima in the definition of $\cQ_n$ proves~\eqref{eq:confidenceonlylower}; neither infimum needs to be attained.

Since $\delta\le1/3$,~\eqref{eq:stronglower} also gives
\[
 \cQ_n(\rho,R,\delta)
 \ge10^{-6}\min\{R,\sqrt{\Gamma/n}\}.
\]
For nonnegative $x,y$, the elementary inequality
\[
 \min\{R,\sqrt{x+y}\}
 \le\sqrt2\max\{\min\{R,\sqrt x\},\min\{R,\sqrt y\}\}
\]
then allows us to combine the last bound with~\eqref{eq:confidenceonlylower}. We obtain
\begin{equation}\label{eq:fullconfidencelower}
 \cQ_n(\rho,R,\delta)\ge\frac{10^{-6}}{\sqrt2}
 \min\left\{R,\sqrt{\frac{\Gamma+\log(1/\delta)}n}\right\}.
\end{equation}
Together with~\eqref{eq:fullconfidenceupper} and the expected-risk bounds, this proves Theorem~\ref{thm:main}.

\subsection{Copy complexity}

We now prove Corollary~\ref{thm:confidence}. Write $\Gamma_\delta=\Gamma+\log(1/\delta)$ and let $c_{\rm Q}=10^{-6}/\sqrt2$, the constant in~\eqref{eq:fullconfidencelower}. Choose a universal $c_0$ with $0<c_0<c_{\rm Q}$, and suppose $\varepsilon\le c_0R$. Since this assumption implies the accuracy restrictions in Lemma~\ref{lem:twopoint}, the two-point argument excludes $n=0$.

If an $n$-copy procedure succeeds uniformly at accuracy $\varepsilon$ and confidence $1-\delta$, then $\cQ_n(\rho,R,\delta)\le\varepsilon$. From~\eqref{eq:fullconfidencelower}, it follows that
\[
 c_{\rm Q}\min\{R,\sqrt{\Gamma_\delta/n}\}\le\varepsilon<c_{\rm Q}R.
\]
The minimum must therefore equal $\sqrt{\Gamma_\delta/n}$, which gives $n\ge c_{\rm Q}^2\Gamma_\delta/\varepsilon^2$. Conversely, the concrete procedure proving~\eqref{eq:fullconfidenceupper} has error at most $16\sqrt{\Gamma_\delta/n}$ with failure probability at most $\delta$. Taking $n=\lceil256\Gamma_\delta/\varepsilon^2\rceil$ proves the upper bound in~\eqref{eq:copy}.

The small value allowed for $c_0$ comes from inverting bounds with conservative universal constants. It does not identify a statistical transition at $\varepsilon=c_0R$. The confidence-radius characterization~\eqref{eq:fullconfidence} itself holds for every $n\ge1$ and every $0<\delta\le1/3$.

\begin{remark}[A classical boundary model]\label{rem:classicalconfidence}
The confidence cost behaves differently in a commuting boundary model. Consider Bernoulli observations with parameter $p\in[0,R]$, where $0<R\le1/2$. If $0<\varepsilon\le R/8$ and $0<\delta\le1/3$, the sample complexity for absolute error at most $\varepsilon$ and failure at most $\delta$ is
\[
 \asymp \frac{R\log(1/\delta)}{\varepsilon^2}.
\]
To prove the upper bound, let $\widehat p$ be the empirical mean. Since $p(1-p)\le R$, Bernstein's inequality gives
\[
 \Prb_p\{|\widehat p-p|>\varepsilon\}
 \le2\exp\{-n\varepsilon^2/[2(R+\varepsilon/3)]\}.
\]
The assumptions $\varepsilon\le R/8$ and $\delta\le1/3$ then give the claimed upper bound.

For the lower bound, put $v=R/2$, $u=2\varepsilon$, and $p_\pm=v\pm u$. Since $u\le R/4$, both parameters lie in $[0,R]$, and their separation is $2u=4\varepsilon$. As in~\eqref{eq:diagonalfidelity}, their squared classical fidelity satisfies
\[
 \sqrt F=\sqrt{v^2-u^2}+\sqrt{(1-v)^2-u^2}
 \ge1-\frac{u^2}{v(1-v)}
 \ge1-\frac{32\varepsilon^2}{3R},
\]
where the last inequality follows from $v(1-v)\ge3R/8$. Squaring the nonnegative lower bound and using $(1-x)^2\ge1-2x$, we obtain $F\ge1-22\varepsilon^2/R$. Since $22\varepsilon^2/R\le11/64<1/2$, we then have $-\log F\le44\varepsilon^2/R$. Applying the testing argument leading to~\eqref{eq:binarycopylower} and~\eqref{eq:logcomparison} proves the lower bound. By comparison, the coherent rotations in Lemma~\ref{lem:twopoint} have fidelity loss of order $\varepsilon^2$, without division by $R$, and therefore yield the quantum confidence cost without a radius prefactor.
\end{remark}

\section{Conclusion and future work}\label{sec:conclusion}

The local minimax risk is governed by $d(\Tr\sqrt\rho)^2+d^2R$, and the confidence radius adds $\log(1/\delta)$ to this coefficient. The bounds hold uniformly over finite dimensions, arbitrary center spectra, and all sample sizes. The matching converse combines square-root perturbations with a weighted trace-norm witness; the upper bound uses an existing universal measurement and a confidence-dependent choice between that estimator and the known center. Equation~\eqref{eq:savingsconfidence} identifies the resulting sample savings, including the regime in which confidence cost dominates dimension.

Several questions remain. First, what is the corresponding center-dependent rate under single-copy or bounded-batch measurements? Existing rank-dependent results~\cite{NZ} show that restricting joint measurements changes global complexity, but do not give the full local spectral dependence. Second, can the same local and confidence guarantees be attained with controlled running time and quantum memory? The streaming methods of Hu \emph{et al.}~\cite{Hu} provide relevant implementation tools.

Two statistical questions also remain open. A procedure without a rank bound may improve the finite-sample remainder while retaining the leading spectral term $\sqrt{d/n}\,S(\sigma)$; the ancilla-dependent bias of the present estimator does not exclude this possibility. Finally, the fixed-rank rate at normalized-projector centers with $r\le d/2$ is implied by existing results (Corollary~\ref{cor:fixedrank}), but the corresponding uniform characterization at arbitrary nonuniform fixed-rank centers remains unresolved.
\section{Declarations of AI use}
Chat GPT 5.6 Sol was used to explore the proof strategy. Chat GPT 6 Astra was used to draft the paper and refine the mathematical arguments. The proofs were independently checked and approved by the author, who takes full responsibility for the scientific content of the paper.
\appendix
\section{Gaussian concentration with the constants used above}\label{app:gaussian}

Let $Z$ be standard Gaussian in $\mathbb R^q$, and let $g$ be
$L_g$-Lipschitz with $L_g>0$. Then
\begin{equation}\label{eq:gaussconc}
 \Prb\{g(Z)\ge\E g(Z)+u\}\le\exp\{-u^2/(2L_g^2)\},
 \qquad u\ge0.
\end{equation}
For $L_g=1$ this is~\cite[Equation (1.22)]{Ledoux}, and the general
case follows by applying it to $g/L_g$. It also follows from the
Gaussian logarithmic-Sobolev
inequality~\cite[Section 2.2, Equation (2.15)]{Ledoux} by the Herbst
argument~\cite[Section 2.3, Corollary 2.6]{Ledoux}. In Section~\ref{sec:upper}, $L_g=\sqrt2$ and
$u=2\sqrt t$, so the failure probability is $e^{-t}$, as used
in~\eqref{eq:fconcentration}.

\section{Eigenvalue variation along a Hermitian matrix path}\label{app:eigen}

The Lidskii--Mirsky inequality used in~\eqref{eq:Sloc} follows from
standard Hermitian perturbation theory~\cite{Bhatia}. We give its short
matrix-path proof to explain why repeated eigenvalues require no
additional assumption.

Put $B=\sigma-\rho$ and $H(t)=\rho+tB$ for $0\le t\le1$, and let
$\lambda_i(t)$ be the ordered eigenvalues of $H(t)$. Since the min--max
principle gives
$|\lambda_i(t)-\lambda_i(s)|\le|t-s|\norm B_{\op}$, these functions
are absolutely continuous and simultaneously differentiable for
almost every $t$. At such a point, their derivatives within an
eigenspace with spectral projection $P$ are the eigenvalues of
$PBP$ restricted to that eigenspace. Consequently, choosing an
orthonormal eigenbasis of each compression gives an orthonormal basis
$(u_i)$ such that $\lambda_i'(t)=\ip{u_i}{Bu_i}$.

Since $|\ip u{Bu}|\le\ip u{|B|u}$ for every unit vector $u$, we then
have
\[
 \sum_i|\lambda_i'(t)|
 \le\sum_i\ip{u_i}{|B|u_i}=\Tr|B|
 =\norm{\sigma-\rho}_1
\]
for almost every $t$. Absolute continuity now gives
\[
 \sum_i|\lambda_i(\sigma)-\lambda_i(\rho)|
 \le\int_0^1\sum_i|\lambda_i'(t)|\,dt
 \le\norm{\sigma-\rho}_1.
\]
The argument uses the compressed perturbation $PBP$ at a repeated
eigenvalue; it does not require differentiable eigenvectors.

\section{The fixed-rank rate at a normalized projector}\label{app:rank}

We record the local minimax rate at a normalized projector under a rank
constraint. This rate is implied by existing results, as explained at
the end of this appendix. The sign family of Section~\ref{sec:lower},
applied to the projector without regularization, also gives the lower
bound directly and with explicit constants. Its members have rank
exactly $r$, but their nonzero eigenvalues need not be equal, so the
construction does not impose the stronger normalized-projector promise
in~\cite{SSW}.

\begin{proposition}[Rank-preserving sign family]\label{prop:rankpreserving}
Let $d\ge2$, $1\le r\le d/2$, $0<R\le1$, and let $P_r$ be a rank-$r$
orthogonal projector. For every procedure using $n\ge1$ copies, there is a
rank-$r$ state $\sigma$ satisfying $T(\sigma,P_r/r)\le R$ such that
\[
 \Prb_\sigma\left\{T(\widehat\sigma,\sigma)
 >10^{-5}\min\{R,\sqrt{rd/n}\}\right\}>0.43.
\]
\end{proposition}

\begin{proof}
Choose a basis in which $P_r=\diag(I_r,0)$. We repeat the construction of
Section~\ref{sec:lower} with $\rho_0$ replaced by $P_r/r$ itself and with
$k=r$ and $m=d-r$; that is, $\mu_i=1/r$ for $i\le r$ and $\mu_i=0$
otherwise. Then $k\le m$, and~\eqref{eq:L} gives $s_0=\sqrt r$ and
$\ell_0=\sqrt{rm}$. The nonzero columns of $A_\Xi=D+aE_\Xi$
in~\eqref{eq:cube} form the matrix with blocks $I_r/\sqrt r$ and
$a\Xi$, so every $\sigma_\Xi$ has rank exactly $r$.

The regularization in Section~\ref{sec:lower} was used only
for~\eqref{eq:Llower} and~\eqref{eq:bsmall}. Here these bounds are
replaced by $\ell_0^2=rm\ge rd/2$ and
\[
 b=a^2rm=(a\ell_0)^2\le10^{-4}R^2,\qquad 1\le h\le1.0001.
\]
No regularization is needed, so the estimate following~\eqref{eq:detrow}
becomes
\[
 T(\sigma_\Xi,P_r/r)\le\frac{a\ell_0+b}{h}\le0.0101R,
\]
and the whole family lies in the prescribed ball. The remaining steps
use only $a\le1/(100\sqrt n)$, $1\le h<2$, $k\le m$, and the
nonnegativity of the weights $\sqrt{\mu_i}$. They therefore apply verbatim: the neighbor
bound~\eqref{eq:neighborTV}, the weighted Assouad
bound~\eqref{eq:Wprob}, the dual-witness bound~\eqref{eq:lossW}, and the
sign-matrix bound~\eqref{eq:signnorm}. As in Section~\ref{sec:lower},
with prior-averaged probability greater than $0.43$ we obtain
\[
 T(\widehat\sigma,\sigma_\Xi)\ge\frac{a\ell_0}{200}
 =\frac1{20000}\min\{R,\sqrt{rm/n}\}
 \ge\frac1{20000\sqrt2}\min\{R,\sqrt{rd/n}\}.
\]
Since $1/(20000\sqrt2)>10^{-5}$, at least one vertex satisfies the
claim.
\end{proof}

\begin{corollary}[Local minimax rate at a normalized projector]
\label{cor:fixedrank}
Let $d\ge2$, $1\le r\le d/2$, $0<R\le1$, and let $n\ge1$ be an integer.
Let $P_r$ be a rank-$r$ orthogonal projector. Define
\[
 \cB^{=r}_R(P_r/r)
 =\{\sigma\in\cD_d:T(\sigma,P_r/r)\le R,
                         \ \operatorname{rank}\sigma=r\},
\]
\[
 \cB^{\le r}_R(P_r/r)
 =\{\sigma\in\cD_d:T(\sigma,P_r/r)\le R,
                         \ \operatorname{rank}\sigma\le r\}.
\]
For either class $\mathcal C$, define $\cR_n(\mathcal C)$ and
$\cQ_n(\mathcal C,\delta)$ as in Section~\ref{sec:model}, replacing
$\cB_R(\rho)$ by $\mathcal C$. The procedure may depend on $P_r,r,R,n$
and may output any density matrix. Then, uniformly in all parameters,
\[
 \cR_n(\mathcal C)\asymp\min\{R,\sqrt{rd/n}\},
 \qquad
 \cQ_n(\mathcal C,1/3)\asymp\min\{R,\sqrt{rd/n}\}.
\]
\end{corollary}

\begin{proof}
Since $\cB^{=r}_R(P_r/r)\subseteq\cB^{\le r}_R(P_r/r)$,
Proposition~\ref{prop:rankpreserving} applies to either class. It gives
the confidence-radius lower bound because $0.43>1/3$, and the
expected-risk lower bound after multiplying the threshold by $0.43$.

For the upper bounds, return $P_r/r$ when $n\le rd/R^2$; both criteria
are then at most $R=\min\{R,\sqrt{rd/n}\}$. When $n>rd/R^2$, we have
$dr/n<R^2\le1$, and the target rate is $\sqrt{dr/n}$. Use the
purification/Hayashi procedure with an $r$-dimensional ancilla, which is
admissible because the rank bound is known.
By~\eqref{eq:knownrankexpected} and $S(\sigma)\le\sqrt r$, its expected
error is at most $\sqrt{dr/n}+dr/n\le2\sqrt{dr/n}$. Markov's inequality
then bounds the probability of an error exceeding $6\sqrt{dr/n}$ by
$1/3$.
\end{proof}

\paragraph{Global consequence.}
Let $1\le r\le d$ and apply Proposition~\ref{prop:rankpreserving} with
rank $r_0=\min\{r,\lfloor d/2\rfloor\}$ and $R=1$. Since $r_0\ge r/3$,
every procedure errs by more than
$5\cdot10^{-6}\min\{1,\sqrt{rd/n}\}$ with probability greater than
$0.43$ at some state of rank at most $r$. This is an alternative proof
of the known global $\Omega(rd/\varepsilon^2)$ lower bound for
rank-at-most-$r$ tomography~\cite{SSW,NZ}.

\paragraph{Relation to existing results.}
Corollary~\ref{cor:fixedrank} is implied by earlier work, up to the
values of the constants. For the lower bound, the proof
of~\cite[Proposition 4.13]{NZ} treats $r\le d/2$ and uses the tomography
guarantee only on the states $\rho_X=\Pi_X/r$, where $\Pi_X$ projects
onto the graph of a $(d-r)\times r$ matrix $X$ with
$\norm X_{\op}\le L\eta$. Here $\eta$ is the accuracy in trace norm and
$L$ is a universal constant. If $s_1,\ldots,s_r$ are the singular values
of $X$, the principal angles between the graph and the range of $P_r$
have tangents $s_i$, and hence
\[
 T(\rho_X,P_r/r)=\frac1r\sum_{i=1}^r\frac{s_i}{\sqrt{1+s_i^2}}
 \le\norm X_{\op}\le L\eta.
\]
This inequality is not stated in~\cite{NZ}. It places their hard family
in $\cB^{=r}_R(P_r/r)$ whenever $L\eta\le R$. Their argument allows a
joint measurement on all $n$ copies and shows that, for $\eta$ below a
universal constant, success probability $2/3$ at accuracy $\eta$ on this
family requires $n\ge c_{\rm NZ}rd/\eta^2$. When $\eta$ is a
sufficiently small universal multiple of $\min\{R,\sqrt{rd/n}\}$, this
requirement fails. Some state of the family then has trace-norm error
above $\eta$ with probability above $1/3$, which gives the lower bounds
of the corollary. For the upper bound, the display at the start
of~\cite[Section 4]{OW}, based on Theorem 1.2 there, gives expected
trace distance $O(\sqrt{rd/n})$ for states of rank at most $r$. Proposition~\ref{prop:rankpreserving} adds a discrete family
and explicit constants.

At the same center, Theorem~\ref{thm:main} gives the unrestricted-ball
rate $\min\{R,\sqrt{(dr+d^2R)/n}\}$. In comparison with
Corollary~\ref{cor:fixedrank}, the additional $d^2R$ term records the
cost of allowing the rank to increase. The corollary concerns normalized
projectors; it does not determine the local minimax rate around an
arbitrary nonuniform fixed-rank center.

\end{document}